\documentclass[12pt, a4paper]{article}
\usepackage[utf8]{inputenc}
\usepackage[T1]{fontenc}
\usepackage{amsmath, amssymb, amsthm}
\usepackage{mathtools}
\usepackage[pagewise]{lineno}
\usepackage{lmodern}
\usepackage[margin=2.5cm]{geometry}
\usepackage{hyperref}
\usepackage{enumitem}
\usepackage{xcolor}
\usepackage{tcolorbox}
\usepackage{booktabs}
\usepackage{array}

\usepackage{microtype}
\usepackage{listings}

\hypersetup{
  colorlinks = true,
  linkcolor  = blue!60!black,
  citecolor  = blue!60!black,
  urlcolor   = blue!60!black
}

\tcbuselibrary{skins, breakable}

\theoremstyle{plain}
\newtheorem{theorem}{Theorem}[section]
\newtheorem{lemma}[theorem]{Lemma} 
\newtheorem{corollary}[theorem]{Corollary}

\newtheorem{definition}[theorem]{Definition}
\newtheorem{remark}[theorem]{Remark}
\newtheorem{example}[theorem]{Example}
\newtheorem{proposition}[theorem]{Proposition}

\newtcolorbox{researchbox}[2][]{
  colback   = blue!4!white,
  colframe  = blue!40!black,
  fonttitle = \bfseries,
  title     = {#2},
  breakable,
  #1
}

\newtcolorbox{notebox}{
  colback  = orange!6!white,
  colframe = orange!60!black,
  fonttitle = \bfseries,
  title    = {Remark},
  breakable
}

\usepackage[textsize=scriptsize]{todonotes}

\newcommand{\F}{\mathbb{F}}
\newcommand{\Fq}{\mathbb{F}_{q}}
\newcommand{\Fqm}{\mathbb{F}_{q^m}}

\newcommand{\Hull}{\operatorname{Hull}}
\newcommand{\Tr}{\operatorname{Tr}}

\newcommand{\rank}{\operatorname{rank}}
\newcommand{\im}{\operatorname{im}}
\newcommand{\Clm}{C_{\lambda,\mu}}
\newcommand{\phiop}{\phi_{\lambda,\mu}}

\newcommand{\PP}{\mathbb{P}}
\newcommand{\ochar}{\operatorname{char}}

\title{%
  \textbf{On the hull of linearized polynomial codes}
}
 \author{\Large \bf Daniele Bartoli\thanks{Department of Mathematics and Computer Science, University of Perugia,  06123 Perugia, Italy;
\texttt{daniele.bartoli@unipg.it}}
\and \Large \bf Giovanni Giuseppe Grimaldi\thanks{Department of Mathematics and Computer Science, University of Perugia,  06123 Perugia, Italy;
\texttt{giovannigiuseppe.grimaldi@unipg.it}}
\and \Large \bf Pantelimon St\u anic\u a\thanks{Applied Mathematics Department, Naval Postgraduate School,
Monterey, CA 93943, USA;
\texttt{pstanica@nps.edu}}}
\date{}

\begin{document}

\maketitle
 \begin{abstract}
Motivated by applications in {which hull dimensions determine complementary Gram ranks, including standard constructions of }entanglement-assisted quantum error-correcting codes (EA-QECCs), we study the hull of two families of $\mathbb{F}_q$-linear codes defined by $q$-polynomial operators over $\mathbb{F}_{q^m}$ using a unified Gram-matrix approach.  For image codes $\mathcal{C}(\boldsymbol{\alpha})=\im \Phi_{\boldsymbol{\alpha}}(x)$, where $\Phi_{\boldsymbol{\alpha}}=\sum_i\alpha_iF_i$ is a linear combination of fixed $q$-polynomials, our master hull--rank theorem establishes $\dim\Hull(\mathcal{C}(\boldsymbol{\alpha})) =\rank(\Phi_{\boldsymbol{\alpha}})-\rank(G(\boldsymbol{\alpha}))${.  For base-field parameters, }the associated Gram matrix {satisfies $G_T=A_T^{\top}G_0A_T$; hence its determinant vanishes exactly when the defining operator is singular}.  Specializing to {$C_{\lambda,\mu}=\im \phi_{\lambda,\mu}$ with $\phi_{\lambda,\mu}(x)=\lambda x+\mu L(x)$ }yields a quadratic {Gram pencil whose determinant detects the singular fibres }in $\mathbb{P}^1(\mathbb{F}_q)${, while the LCD condition on a singular fibre is determined by the rank difference in the master theorem}.  In a parallel setting for $\mathbb{F}_{q^m}$-linear rank-distance codes equipped with the Delsarte inner product, an analogous {Gram matrix }directly determines the hull dimension.  For the Frobenius twist $L=X^{q^k}$, $d=\gcd(k,m)$, the circulant structure {yields an explicit factorization of the discriminant over a splitting field and shows that its base-field roots are precisely the singular parameters.  At such a root }the hull dimension {is $0$ or $d$, according to the multiplicity of the corresponding Gram zero.  Over the full parameter space $\mathbb{P}^1(\mathbb{F}_{q^m})$, the adjoint description gives $\Hull(C_{\lambda,\mu})=\im\phi_{\lambda,\mu}\cap\ker\phi_{\lambda,\mu}^{\dagger}$, with dimension between $0$ and $d$ on the singular locus.  A }worked example over $\mathbb{F}_{64}$ and {the SageMath computations }in the appendix {provide exhaustive verification for that fixed field}.
\end{abstract}

\noindent \thanks{\textit{2020 MSC:} 94B05; 11T06; 11T71}

\noindent \thanks{{\textit{Keywords:} Finite fields; linearized polynomials; EA-QECCs; RD codes}}


\section{Introduction}

Let $q$ be a $p$ (prime) power and let $\Fqm$ denote the degree-$m$ extension of $\Fq$.
The \emph{hull} of a linear code $C$ is
\[
  \Hull(C) \;=\; C \cap C^{\perp},
\]
where $C^{\perp}$ is the dual with respect to a non-degenerate bilinear form.
This notion was introduced by Assmus and Key~\cite{AK92} in the context of
classifying finite projective planes, where the hull of the incidence code carries
geometric information about the plane.
The hull dimension is a fundamental invariant of a linear code: $\Hull(C)=\{0\}$
characterizes \emph{linear complementary dual} (LCD) codes, $C\subseteq C^{\perp}$
characterizes self-orthogonal codes, and $C=C^{\perp}$ characterizes self-dual codes.

LCD codes were introduced by Massey~\cite{Massey92}, who showed that they are
asymptotically good and provide an optimal solution to the two-user binary adder channel.
Sendrier~\cite{Sendrier04} subsequently proved that LCD codes meet the asymptotic
Gilbert--Varshamov bound.
Interest in the hull was renewed by the work of Carlet and
Guilley~\cite{CG16}, who demonstrated that binary LCD codes offer
effective resistance against side-channel and fault injection attacks in
hardware implementations of symmetric-key cryptography.
This observation initiated an intensive study of LCD codes and, more broadly, of
codes with prescribed hull dimension over finite fields~\cite{CMTQP18,CLM19,GJG18}.

The relationship between the hull and quantum error correction is {also useful, but the precise ebit count depends on the EAQECC construction}.  Calderbank and Shor~\cite{CS96} and, independently, Steane~\cite{Steane96} showed that self-orthogonal codes give rise to quantum stabilizer codes via the CSS construction.  {In the }entanglement-assisted {setting of }Brun, Devetak, and Hsieh~\cite{BDH06}, {for a classical $[n,k]$ code $C\subseteq\Fq^n$ with Euclidean dual and parity-check matrix $H$, one standard construction uses
$ c=\rank(HH^{\top})=n-k-\dim\Hull(C)$ ebits; the corresponding generator-matrix Gram rank is $\rank(GG^{\top})=k-\dim\Hull(C)$.  Thus the hull determines the relevant complementary Gram rank once the classical code, duality, and EAQECC construction are specified; it is not, in general, itself }the number of {required ebits}.

In the rank-metric setting, rank-distance codes were introduced by Delsarte~\cite{Delsarte78} and, independently, by Gabidulin~\cite{Gabidulin85},
who constructed the family of maximum rank distance (MRD) codes now bearing his name. These codes can be  defined using $q$-polynomials (linearized polynomials), and their duality theory uses the Delsarte inner product on coefficient vectors.
The hull of a rank-metric code with respect to the Delsarte dual has received increasing attention in recent years~\cite{HJ26}, prompted by the same quantum
coding applications that drove the Hamming-metric hull theory.

This paper studies the hull of linear codes defined by $q$-polynomial operators over $\Fqm$, using Gram matrices in two parallel settings. For image codes of the form $\mathcal{C}(\boldsymbol\alpha)=\im\Phi_{\boldsymbol\alpha}(x)$
with $\Phi_{\boldsymbol\alpha}=\sum_i \alpha_iF_i$, the central object is the \emph{quadratic Gram form}
\[
  G(\boldsymbol\alpha) \;=\; \sum_{i,j}\alpha_i\alpha_j\,\Gamma_{ij},
  \qquad \Gamma_{ij}=[\Tr(F_i(e_s)F_j(e_t))]_{s,t},
\]
and the hull dimension is shown to equal
$\rank(\Phi_{\boldsymbol\alpha})-\rank(G(\boldsymbol\alpha))$.
In the two-parameter family $C_{\lambda,\mu}=\im \phi_{\lambda, \mu}$, where $\phi_{\lambda, \mu}(x)=\lambda x + \mu L(x)$, this yields
a quadratic matrix pencil whose determinant {detects singular fibres of the base-field operator pencil; the LCD condition on a singular fibre is determined by the rank difference in the master theorem}. For rank-distance codes
$\mathcal{C}=\langle X,F_1,\ldots,F_k\rangle_{\Fqm}$,
a separate but analogous generator-Gram-matrix argument yields a $(k\times k)$ Gram matrix over $\Fqm$ whose rank determines the Delsarte hull.
The two families are therefore governed by the same Gram-matrix philosophy, although not by a single literal specialization theorem.

The {elementary identity relating the radical of a restricted non-degenerate bilinear form to the rank of its Gram matrix is standard.  The contribution here is not that identity itself, but its systematic use for linearized-polynomial image codes together with the explicit operator factorization of the Gram pencil, the resulting singularity stratification for the two-parameter family, the splitting-field analysis of the Frobenius-twist circulants, and the separate Delsarte generator-Gram formulation for rank-distance codes.
}
The paper is organized as follows.
Section~\textup{\ref{sec:background}} collects the necessary algebraic background.
Section~\textup{\ref{sec:general}} introduces the general linear code construction, establishing the quadratic Gram form, the hull stratum decomposition, and the master hull--rank theorem.
Section~\textup{\ref{sec:rdcodes}} specializes to $\Fqm$-linear rank-distance codes,
deriving explicit hull criteria from the corresponding generator Gram matrix.
Section~\textup{\ref{sec:twoparam}} then turns to the two-parameter family $C_{\lambda,\mu}$,
viewed as the image/support code attached to the operator $\phi_{\lambda,\mu}$, and analyzes the resulting pencil discriminant and its roots.
Section~\textup{\ref{sec:frobenius}} treats the Frobenius twist $L=X^{q^k}$ in detail.
Section~\textup{\ref{sec:conclusions}} collects conclusions and open problems.

\section{Algebraic background}
\label{sec:background}

A \emph{$q$-polynomial} (or \emph{linearized polynomial}) over $\Fqm$ is a polynomial of
the form $L = \sum_{i=0}^{s} a_i X^{q^i}$ with $a_i \in \Fqm$. The \emph{$q$-rank} of $L$ is $\rank(L)\coloneqq\dim_{\Fq}\im L(x)$, where $L(x)$ is the $\Fq$-linear endomorphism $x \in \F_{q^m} \mapsto \sum_{i=0}^{s}a_i x^{q^i} \in \F_{q^m}$; $L$ is called \emph{invertible} if $\rank(L)=m$. The composition of two $q$-polynomials is again a $q$-polynomial, making the set of $q$-polynomials modulo $X^{q^m}-X$ into the
\emph{Ore polynomial ring} $\Fqm[X;\sigma]$, where $\sigma\colon x\mapsto x^q$ is the Frobenius automorphism.   The simplest non-trivial example is the Frobenius twist $L=X^{q^k}$, $1\le k\le m-1$, which is invertible. {For $\eta\in\Fqm$ and $d=\gcd(k,m)$, the equation $\alpha^{q^k}=\eta\alpha$ has a solution $\alpha\in\Fqm^*$ if and only if $\eta^{(q^m-1)/(q^d-1)}=1$.}

The ambient inner product throughout is the \emph{trace bilinear form}
\[
  \langle a,b\rangle \;=\; \Tr_{\Fqm/\Fq}(ab), \qquad a,b\in\Fqm,
\]
where $\Tr_{\Fqm/\Fq}(x)=x+x^q+\cdots+x^{q^{m-1}}$.  This form is non-degenerate,
symmetric, and $\Fq$-bilinear; it identifies $\Fqm$ with its own $\Fq$-linear dual.
The \emph{adjoint} of an $\Fq$-linear operator $T$ with respect to this form is the
unique operator $T^\dagger$ satisfying $\langle T(a),b\rangle=\langle a,T^\dagger(b)\rangle$
for all $a,b\in\Fqm$; for $L=\sum_i a_i X^{q^i}$ one has
$L^\dagger=\sum_i a_i^{q^{m-i}}X^{q^{m-i}}$.

Fix a $\Fq$-basis $\{e_1,\ldots,e_m\}$ of $\Fqm$.  The \emph{Gram matrix} of an
$\Fq$-linear operator $T\colon\Fqm\to\Fqm$ is
\[
  G_T \;=\; \bigl[\Tr(T(e_i)\,T(e_j))\bigr]_{1\le i,j\le m}
  \;\in\;\mathcal{M}_{m\times m}(\Fq).
\]
If $T:\F_{q^m} \longrightarrow \F_{q^m}$ is an $\F_q$-linear operator and $\{e_1,\ldots,e_m \}$ is an $\F_q$-basis of $\F_{q^m}$, then putting $T(e_s)=\sum_{i=1}^m A_{is}e_i$ we get {$(G_T)_{s,t}=\sum_{i=1}^m\sum_{j=1}^m A_{is}A_{jt}\Tr(e_i e_j)$}. Therefore, if $G_0=[\Tr(e_s e_t)]_{1 \leq s,t \leq m}$, then $G_T=A^{\top} G_0 A$ where $A$ is the matrix associated to the operator $T$ with respect to the basis $\{ e_1,\ldots,e_m \}$. In particular, $\det(G_T)=\det(G_0) \det(A)^2$ and since the trace form is non-degenerate then the matrix $G_0$ is invertible and the matrix $G_T$ is non-singular if and only if $A$ is non-singular, or equivalently $T$ is bijective.

Note that the matrix $G_T$ records the pullback of the trace form to the domain of $T$.
If $C=\im(T)\subseteq\Fqm$, then
\begin{equation}
  \label{eq:hullrank-basic}
  \dim_{\Fq}\Hull(C)
  \;=\;
  \rank(T)-\rank(G_T).
\end{equation}
Indeed, since $\ker(G_T)=T^{-1}(C^{\perp})$, then
$\dim\ker(G_T)=\dim\ker(T)+\dim(C\cap C^{\perp})$; since
$\rank(G_T)=m-\dim\ker(G_T)$ and $\rank(T)=m-\dim\ker(T)$, the identity follows.
This relation is the engine of the image-code results below.

A \emph{normal basis} of $\Fqm/\Fq$ is a basis of the form
$\{\beta,\beta^q,\ldots,\beta^{q^{m-1}}\}$ for some $\beta\in\Fqm$ (a \emph{normal element}).
Normal bases always exist by the Normal Basis Theorem, and they are particularly convenient
because the trace form satisfies $\Tr(\beta^{q^i}\beta^{q^j})=\Tr(\beta^{1+q^{j-i}})$,
a function only of $j-i\pmod{m}$.  Consequently, when an operator $T$ commutes with
Frobenius, the corresponding Gram matrix $G_T$ becomes circulant in a normal basis.
This will be exploited in Section~\textup{\ref{sec:frobenius}}.

\section{Square rank-metric  codes and their codewords}
\label{sec:general}

Rank-metric codes are subsets of the vector space $\mathbb{F}_q^{m \times n}$ equipped with a distance function 
$$
(A,B) \in \mathbb{F}_q^{m \times n} \times \mathbb{F}_q^{m \times n} \mapsto \rank(A-B) \in \{ 0, \ldots, \min\{m,n\} \}.
$$
Linear codes that achieve the maximum possible cardinality for a given minimum distance $d$, reaching the Singleton-like bound $|\mathcal{C}| \leq q^{\max\{n,m\}(\min\{n,m\}-d+1)}$, are called Maximum Rank-Distance (MRD) codes \cite{Delsarte78, Gabidulin85}. These codes have drawn significant attention in recent decades due to their critical applications in network coding and cryptography \cite{koetter, wachterzeh}. When $m=n$, these square rank-metric codes can be naturally represented as sets of $\mathbb{F}_q$-linearized polynomials over $\mathbb{F}_{q^n}$ \cite[Sections 2.2 and 2.3]{articolosequenze}. From now on, we will adopt this polynomial description for the remainder of the paper.

 Let $F_0, F_1, \ldots, F_k \in \Fqm[X;\sigma]$ be fixed $q$-polynomials, and let $\mathcal{C}=\langle F_0, F_1, \ldots, F_k \rangle_{\mathbb{F}_{q^m}}$ be the corresponding $\mathbb{F}_{q^m}$-linear rank-metric code. Given a vector $\boldsymbol{\alpha} = (\alpha_0, \alpha_1, \ldots, \alpha_k) \in \Fqm^{k+1}$, we define the linearized polynomial $$\Phi_{\boldsymbol{\alpha}} := \sum_{i=0}^{k} \alpha_i F_i$$ and let $\mathcal{C}(\boldsymbol{\alpha}) := \im \Phi_{\boldsymbol{\alpha}}(x)$. By representing $\Phi_{\boldsymbol{\alpha}}$ as the transpose of the standard matrix associated with $\Phi_{\boldsymbol{\alpha}}(x)$ (with respect to a fixed $\mathbb{F}_q$-basis of $\mathbb{F}_{q^m}$), the space $\mathcal{C}(\boldsymbol{\alpha})$ is exactly the rank-metric support $\textrm{supp}(\Phi_{\boldsymbol{\alpha}})$, see e.g. \cite{linearcutting}.


Note that we may also consider the  $\mathbb{F}_q$-linear rank-metric code $\mathcal{C}$ spanned by the same polynomials, $\langle F_0, F_1, \ldots, F_k \rangle_{\mathbb{F}_q}$; the intended field of linearity will be clear from the context. Such a code is a subcode of $\langle F_0, F_1, \ldots, F_k \rangle_{\mathbb{F}_{q^m}}$.


  With respect to a fixed $\Fq$-basis
$\{e_1,\ldots,e_m\}$ of $\Fqm$, define the \emph{Gram matrix}
\begin{equation}
  \label{eq:gramgeneral}
  G(\boldsymbol{\alpha}) \;\coloneqq\;
  \bigl[\Tr(\Phi_{\boldsymbol{\alpha}}(e_s)\,\Phi_{\boldsymbol{\alpha}}(e_t))\bigr]_{s,t}
  \;\in\;\mathcal{M}_{m\times m}(\Fq).
\end{equation}
When $\boldsymbol{\alpha}\in\Fq^{k+1}$, the $\Fq$-linearity of the trace gives the
quadratic decomposition
\begin{equation}
  \label{eq:gramdecomp}
  G(\boldsymbol{\alpha}) \;:=\; \sum_{i=0}^{k}\sum_{j=0}^{k} \alpha_i\,\alpha_j\,\Gamma_{ij},
\end{equation}
where $\Gamma_{ij}\in\mathcal{M}_{m\times m}(\Fq)$ has entries
$(\Gamma_{ij})_{st} = \Tr(F_i(e_s)\,F_j(e_t))$.  These structure matrices depend only
on the generators $\{F_i\}$ and the basis; since $\Gamma_{ij}^\top=\Gamma_{ji}$,
the expression~\eqref{eq:gramdecomp} is a symmetric quadratic form in $\boldsymbol{\alpha}$
with values in $\mathcal{M}_{m\times m}(\Fq)$.  For $\boldsymbol{\alpha}\in\Fqm^{k+1}$
the decomposition~\eqref{eq:gramdecomp} need not hold, since
$\Tr(A\,B)\ne A\,\Tr(B)$ for $A\notin\Fq$ in general; however, the Gram
matrix~\eqref{eq:gramgeneral} and the hull--rank theorem below remain valid.

\begin{remark}
If {$T:\F_{q^m}\to\F_{q^m}$ is bijective, then $\im(T)=\F_{q^m}$.  Since }the trace form is non-degenerate{, $(\F_{q^m})^{\perp}=\{0\}$, and therefore $\Hull(\im T)=\{0\}$.  Thus every bijective operator in the image-code construction gives an LCD code.  Equivalently, if $A_T$ is the matrix of $T$, then $G_T=A_T^{\top}G_0A_T$ is nonsingular whenever $T$ is bijective}.
\end{remark}

\begin{theorem}[Master hull--rank theorem]
\label{thm:master}
For every $\boldsymbol{\alpha}\in\Fqm^{k+1}$,
\[
  \dim_{\Fq}\Hull\bigl(\mathcal{C}(\boldsymbol{\alpha})\bigr)
  \;=
  \rank\bigl(\Phi_{\boldsymbol{\alpha}}\bigr)
  \;-
  \rank\bigl(G(\boldsymbol{\alpha})\bigr).
\]
Equivalently,
\[
  \dim_{\Fq}\Hull\bigl(\mathcal{C}(\boldsymbol{\alpha})\bigr)
  \;=
  \dim_{\Fq}\mathcal{C}(\boldsymbol{\alpha})
  \;-
  \rank\bigl(G(\boldsymbol{\alpha})\bigr).
\]
In particular, $\mathcal{C}(\boldsymbol{\alpha})$ is LCD if and only if
$\rank\bigl(G(\boldsymbol{\alpha})\bigr)=\rank\bigl(\Phi_{\boldsymbol{\alpha}}\bigr)$, and
self-orthogonal if and only if $G(\boldsymbol{\alpha})=0$.
{In particular, if }$\Phi_{\boldsymbol{\alpha}}$ is bijective, then {$\mathcal{C}(\boldsymbol{\alpha})=\F_{q^m}$ and is automatically LCD; in this case $G(\boldsymbol{\alpha})$ is necessarily nonsingular}.
\end{theorem}

\begin{proof}
The first identity is~\eqref{eq:hullrank-basic} applied to the operator $T=\Phi_{\boldsymbol{\alpha}}$.
The second follows from the equality
$\dim_{\Fq}\mathcal{C}(\boldsymbol{\alpha})=\rank(\Phi_{\boldsymbol{\alpha}})$.
If $G(\boldsymbol{\alpha})=0$, then every pair of vectors in
$\mathcal{C}(\boldsymbol{\alpha})$ is trace-orthogonal, so the code is self-orthogonal.
Conversely, if the code is self-orthogonal then the pulled-back form vanishes on the
whole domain, hence $G(\boldsymbol{\alpha})=0$.  The LCD criterion is exactly the
condition that the right-hand side be $0$, namely
$\rank(G(\boldsymbol{\alpha}))=\rank(\Phi_{\boldsymbol{\alpha}})${. If $\Phi_{\boldsymbol{\alpha}}$ is bijective, then both ranks equal $m$, consistently with $G(\boldsymbol{\alpha})=A_{\Phi}^{\top}G_0A_{\Phi}$ and $\det G(\boldsymbol{\alpha})\ne0$}.
\end{proof}

Since $G(t\boldsymbol{\alpha})=t^2G(\boldsymbol{\alpha})$ for any scalar $t\in \Fq^*$, the rank of $G(\boldsymbol{\alpha})$, and hence the hull dimension, is constant on $\Fq^*$-orbits in $\Fqm^{k+1}\setminus\{\mathbf{0}\}$.  Note that this is $\Fq^*$-scaling, not $\Fqm^*$-scaling, so the natural parameter space is the orbit space $(\Fqm^{k+1}\setminus\{\mathbf{0}\})/\Fq^*$, which is finer than $\PP^k(\Fqm)$.

\begin{definition}
\label{def:strata}
Let $\mathcal{P} = (\Fqm^{k+1}\setminus\{\mathbf{0}\})/\Fq^*$ denote the space of
$\Fq^*$-orbits in $\Fqm^{k+1}\setminus\{\mathbf{0}\}$.
For $h\in\{0,1,\ldots,m\}$, the \emph{$h$-th hull stratum} is
\[
  \mathcal{S}_h \;=\;
  \bigl\{[\boldsymbol{\alpha}]\in\mathcal{P} \;\colon\;
  \dim_{\Fq}\Hull(\mathcal{C}(\boldsymbol{\alpha}))=h\bigr\}.
\]
The strata $\{\mathcal{S}_h\}_{h=0}^m$ partition $\mathcal{P}$.
On the bijective locus
\[
  \mathcal{P}^{\mathrm{bij}}\;=\;\{[\boldsymbol{\alpha}]\in\mathcal{P}:\Phi_{\boldsymbol{\alpha}}\text{ is bijective}\},
\]
{every point is LCD.  More generally, }the LCD locus is
\[
{\mathcal{S}_0=}\bigl\{[\boldsymbol{\alpha}]\in\mathcal{P}:\Hull(\mathcal{C}(\boldsymbol{\alpha}))=\{0\}\bigr\},
\]
{and $\mathcal{P}^{\mathrm{bij}}\subseteq\mathcal{S}_0$.  The determinantal locus $\det G(\boldsymbol{\alpha})=0$ detects singularity of the defining operator when the matrix factorization $G=A_{\Phi}^{\top}G_0A_{\Phi}$ applies; it is therefore not a ``non-LCD locus'' on the bijective set.
}In the general parameter space, the LCD condition is the rank equality
$\rank(G(\boldsymbol{\alpha}))=\rank(\Phi_{\boldsymbol{\alpha}})$ from Theorem~\textup{\ref{thm:master}}.
The self-orthogonal locus is
$\{[\boldsymbol{\alpha}]: G(\boldsymbol{\alpha})=0\} \subseteq \mathcal{S}_{\rank(\Phi_{\boldsymbol{\alpha}})}$;
a self-orthogonal code $\mathcal{C}(\boldsymbol{\alpha})$ satisfies
$\Hull(\mathcal{C}(\boldsymbol{\alpha}))=\mathcal{C}(\boldsymbol{\alpha})$, so its hull
dimension equals $\dim_{\Fq}\mathcal{C}(\boldsymbol{\alpha})=\rank(\Phi_{\boldsymbol{\alpha}})$.
{Since a self-orthogonal code satisfies
$\mathcal{C}(\boldsymbol{\alpha})\subseteq\mathcal{C}(\boldsymbol{\alpha})^\perp$,
its dimension is at most $\lfloor m/2\rfloor$.  More generally,
$\Hull(\mathcal{C})$ is totally isotropic for every code $\mathcal{C}$, and therefore
$\dim_{\Fq}\Hull(\mathcal{C})\le\lfloor m/2\rfloor$.}
\end{definition}

The \emph{universal null space} $V^*=\bigcap_{i,j}\ker(\Gamma_{ij})$ satisfies
$V^*\subseteq\ker G(\boldsymbol{\alpha})$ for every $\boldsymbol{\alpha}${.  Combining this inclusion with the master hull--rank formula gives the corrected }parameter-free {estimate
}\[
  \dim\Hull \bigl(\mathcal{C}(\boldsymbol{\alpha})\bigr)
  \ge \dim V^*-\dim{\ker(\Phi}_{\boldsymbol{\alpha}}).
\]
{In particular, the stronger bound $\dim\Hull\ge\dim V^*$ is valid whenever $\Phi_{\boldsymbol{\alpha}}$ is injective.
}When $V^*=\{0\}$, the hull can in principle be made trivial by an appropriate choice
of parameters.

\section{The hull of rank-distance codes}
\label{sec:rdcodes}

In this section, we examine rank-metric codes in their full generality. 

We begin by noting that a rank-metric code composed of $m \times m$ matrices is non-degenerate if and only if it contains a codeword of full rank $m$, which algebraically corresponds to the existence of an invertible polynomial. Therefore, up to linear equivalence, we can safely assume that the polynomial $X$ is among the generators of any non-degenerate code.

Fixing a set of $q$-polynomials $F_1, \ldots, F_k \in \Fqm[X;\sigma]$ with zero $X$-coefficient, and setting $F_0 = X$, we define the \emph{rank-distance code} as the $\Fqm$-linear span,
\begin{equation}
  \label{eq:rdcode}
  \mathcal{C} \;=\; \langle X, F_1, \ldots, F_k \rangle_{\Fqm}
  \;\subseteq\; \Fqm[X;\sigma]/(X^{q^m}-X).
\end{equation}
This code constitutes an $\Fqm$-vector space of dimension at most $k+1$, endowed with the rank distance  $d(f,g) = \rank (f-g)$.

 The \emph{Delsarte dual} of $\mathcal{C}$ is typically defined using a trace form on the coefficients. However, since $\mathcal{C}$ is an $\Fqm$-linear space, the trace in the Delsarte pairing is redundant, and the condition simplifies to the $\Fqm$-bilinear inner product,
\[
  \mathcal{C}^\perp \;=\;
  \Bigl\{g \;\colon\; \sum_{\ell=0}^{m-1}g_\ell h_\ell=0 \;\;\forall\,h\in\mathcal{C}\Bigr\},
\]
where $g_\ell,h_\ell\in\Fqm$ are the $q$-polynomial coefficients.

The rank-distance setting is formally parallel to the image-code setting, but it is
not a literal specialization of Theorem~\textup{\ref{thm:master}}.  Here the code lives in the
$\Fqm$-vector space of $q$-polynomials, and the relevant bilinear form is the Delsarte
coefficient pairing.  The hull is therefore controlled directly by the Gram matrix of a
chosen generating set for that pairing.
The key structural observation is that the generator Gram matrix $\mathbf{M}$ of $\mathcal{C}$,
defined by
\begin{equation}
  \label{eq:rdgramsystem}
  \mathbf{M}_{ij} \;=\; \sum_{\ell=0}^{m-1}(F_i)_\ell(F_j)_\ell,
  \qquad 0\le i,j\le k,
\end{equation}
has a block structure determined by the presence of $F_0=X$.
Since $F_0=X$ has coefficient vector $(1,0,\ldots,0)^\top$ and each $F_j$ ($j\ge 1$)
has zero constant term, one has $\mathbf{M}_{00}=1$ and $\mathbf{M}_{0j}=0$ for $j\ge 1$.
Hence
\begin{equation}
  \label{eq:Mblock}
  \mathbf{M} \;=\; \begin{pmatrix}1&\mathbf{0}^\top\\\mathbf{0}&M\end{pmatrix},
  \qquad
  M_{jj'} \;=\; \sum_{\ell=1}^{m-1}(F_j)_\ell(F_{j'})_\ell,
\end{equation}
for $j,j'\in\{1,\ldots,k\}$.  The block $1$ forces $\alpha_0=0$ in every solution of
$\mathbf{M}\boldsymbol{\alpha}^\top=\mathbf{0}$, so the hull is determined entirely
by the $k\times k$ matrix $M$ over $\Fqm$.

\begin{theorem}
\label{thm:rdgeneral}
Let $\mathcal{C}=\langle X,F_1,\ldots,F_k\rangle_{\Fqm}$ with $X,F_1,\ldots,F_k$
$\Fqm$-linearly independent and each $F_j$ having zero constant term, and let $M$
be the $k\times k$ matrix over $\Fqm$ with entries
$M_{jj'}=\sum_{\ell=1}^{m-1}(F_j)_\ell(F_{j'})_\ell$.  Then:
\begin{enumerate}[label=\textup{(\roman*)}]
  \item $\dim_{\Fqm}\Hull(\mathcal{C}) = k - \rank_{\Fqm}(M)$.
  \item $\mathcal{C}$ is LCD if and only if $\det(M)\neq 0$.
  \item $\langle F_1,\ldots,F_k\rangle_{\Fqm}$ is self-orthogonal with respect to the
        Delsarte pairing if and only if $M=0$.
  \item $\mathcal{C}$ is never self-dual, and is never self-orthogonal as a whole.
\end{enumerate}
\end{theorem}
 
\begin{proof}
Since $\mathcal{C}$ is $\Fqm$-linear and the Delsarte pairing $\langle f,g\rangle=\sum_\ell f_\ell g_\ell$ is $\Fqm$-bilinear, a general element $H=\alpha_0 X+\sum_{j=1}^k\alpha_j F_j\in\mathcal{C}$ belongs to $\mathcal{C}^\perp$ if and only if $\langle H,G\rangle=0$ for every $G\in\mathcal{C}$.  By $\Fqm$-bilinearity this is equivalent to $\langle H,F_i\rangle=0$ for each generator $F_i\in\{X,F_1,\ldots,F_k\}$, which is the linear system $\mathbf{M}\boldsymbol{\alpha}^\top=\mathbf{0}$,
where $\boldsymbol{\alpha}=(\alpha_0,\alpha_1,\ldots,\alpha_k)$.

By~\eqref{eq:Mblock}, the block structure of $\mathbf{M}$ yields $\langle H,X\rangle=\alpha_0$ and $\langle H,F_j\rangle=\sum_{j'=1}^k M_{jj'}\alpha_{j'}$ for $j\ge 1$.
Hence $\mathbf{M}\boldsymbol{\alpha}^\top=\mathbf{0}$ forces $\alpha_0=0$ and
$M(\alpha_1,\ldots,\alpha_k)^\top=\mathbf{0}$.
Since $X,F_1,\ldots,F_k$ are $\Fqm$-linearly independent, the map $\boldsymbol{\alpha}\mapsto H$ is an isomorphism $\Fqm^{k+1}\to\mathcal{C}$,
so
\begin{align*}
  \dim_{\Fqm}\Hull(\mathcal{C})
  \;&=\;\dim_{\Fqm}\ker\mathbf{M}
  \;=\;(k+1)-\rank(\mathbf{M})\\
  \;&=\;(k+1)-(1+\rank(M))
  \;=\;k-\rank(M),
\end{align*}
giving~(i).  Part~(ii) follows: $\Hull(\mathcal{C})=\{0\}$ iff $\rank(M)=k$ iff
$\det(M)\ne 0$.

For part~(iii), an element $H=\sum_{j=1}^k \alpha_j F_j\in\langle F_1,\ldots,F_k\rangle_{\Fqm}$
(so $\alpha_0=0$) satisfies $\langle H,F_j\rangle=0$ for all $j$ iff $M\boldsymbol{\alpha}'=0$,
and $\langle H,X\rangle=\alpha_0=0$ automatically.  Hence
$\langle F_1,\ldots,F_k\rangle_{\Fqm}\subseteq\mathcal{C}^\perp$ iff $M=0$, proving~(iii).

For part~(iv), we note that $\langle X,X\rangle=\sum_\ell(X)_\ell^2=1\ne 0$, so $X\notin\mathcal{C}^\perp$,
which means no element of $\mathcal{C}$ with $\alpha_0\ne 0$ lies in $\mathcal{C}^\perp$.
In particular $\mathcal{C}\not\subseteq\mathcal{C}^\perp$, so $\mathcal{C}$ is neither
self-orthogonal nor self-dual.
\end{proof}

Theorem~\textup{\ref{thm:rdgeneral}} subsumes and completes the case analyses for small $k$.


\begin{corollary} 
\label{cor:rdk1}
If $\mathcal{C}=\langle X,f\rangle_{\Fqm}$ with $f=\sum_{i=1}^{m-1}f_i X^{q^i}$,
setting $\sigma_f=\sum_{i=1}^{m-1}f_i^2$,
\[
  \Hull(\mathcal{C}) \;=\;
  \begin{cases}
    \langle f\rangle_{\Fqm} & \text{if } \sigma_f = 0,\\
    \{0\} & \text{if } \sigma_f \neq 0.
  \end{cases}
\]
If $\mathcal{C}=\langle X,f,g\rangle_{\Fqm}$ with $f=\sum_{i=1}^{m-1}f_i X^{q^i}$
and $g=\sum_{i=2}^{m-1}g_i X^{q^i}$,  setting
\[
  a \;=\; \sum_{i=1}^{m-1}f_i^2, \qquad
  b \;=\; \sum_{i=1}^{m-1}f_i g_i, \qquad
  c \;=\; \sum_{i=2}^{m-1}g_i^2,
\]
then, according to the rank of $M=\bigl(\begin{smallmatrix}a&b\\b&c\end{smallmatrix}\bigr)$:
\begin{itemize}[leftmargin=2em]
  \item if $ac - b^2 \neq 0$, then $\mathcal{C} \cap \mathcal{C}^{\perp} = \{\mathbf{0}\}$;
  \item if $ac - b^2 = 0$ and $a \neq 0$ or $b\neq 0$,
        then $\mathcal{C} \cap \mathcal{C}^{\perp} = \langle{-b f + a g}\rangle_{\Fqm}$;
  \item if $a = b = c = 0$,
        then $\mathcal{C} \cap \mathcal{C}^{\perp} = \langle f, g \rangle_{\Fqm}$.
\end{itemize}
\end{corollary}

\begin{proof}
The hull dimension formula in Theorem~\textup{\ref{thm:rdgeneral}}(i) gives $\dim_{\Fqm}\Hull(\mathcal{C}) = 2 - \rank(M)$ in the two-generator case.
The three cases correspond to $\rank(M)=2$, $1$, $0$ respectively.
When $\rank(M)=1$, the kernel of $M$ acting on $(\alpha_1,\alpha_2)^\top$ is spanned
by $(-b,a)^\top$ (or $(-c,b)^\top$ if $a=b=0$), giving the hull generator
$-bf+ag \in \mathcal{C}$.
\end{proof}

\begin{remark}
For codes not containing $X$, the condition $a=b=c=0$ means that
$\langle f,g\rangle_{\Fqm}$ is self-orthogonal with respect to the Delsarte pairing. To conclude self-duality one must in addition verify the ambient dimension condition appropriate to the chosen bilinear space.
\end{remark}

\begin{remark}
\label{rem:nondegen}
The condition $X\notin\mathcal C$ is not equivalent to $\mathcal C$ being degenerate. Indeed, the property $X\in\mathcal C$ is not invariant under code equivalence, whereas degeneracy is. In fact, $\mathcal C$ is degenerate if and only if it has a nontrivial common kernel, namely
\[
\dim_{\F_q}\Bigl(\bigcap_{h\in\mathcal C}\ker(h)\Bigr)>0.
\]
\end{remark}

\section{The two-parameter code family \texorpdfstring{$C_{\lambda,\mu}$}{C(lambda,mu)}}
\label{sec:twoparam}


In this section we investigate the family $\Clm$ from the rank-distance perspective. More precisely, starting from the $\mathbb{F}_q$-rank-metric code $\langle F_0,F_1 \rangle_{\mathbb{F}_q}$, where 
$F_0=X$, $F_1=L$ for a fixed $q$-polynomial $L$, we view $\Clm$ as the image/support code naturally attached to the rank-metric object generated by the operator $\phiop(x)=\lambda x+\mu L(x)$. Thus, for any $\lambda,\mu \in \mathbb{F}_q$, we consider
\[
  \Clm \;=\; \mathcal{C}(\lambda,\mu) \;=\; \im(\phiop) \;\subseteq\;\Fqm.
\]
The three structure matrices $\Gamma_{00},\Gamma_{01},\Gamma_{11}$ of the general
theory specialize to
\begin{align}
  (G_0)_{ij} &\;=\; \Tr_{\Fqm/\Fq}(e_i\,e_j) \;=\; (\Gamma_{00})_{ij}, \label{eq:G0}\\
  (G_1)_{ij} &\;=\; \Tr_{\Fqm/\Fq}(e_i\,L(e_j)) + \Tr_{\Fqm/\Fq}(L(e_i)\,e_j)
               \;=\; (\Gamma_{01}+\Gamma_{10})_{ij}, \label{eq:G1}\\
  (G_2)_{ij} &\;=\; \Tr_{\Fqm/\Fq}(L(e_i)\,L(e_j)) \;=\; (\Gamma_{11})_{ij}, \label{eq:G2}
\end{align}
and the quadratic decomposition~\eqref{eq:gramdecomp} becomes
\begin{equation}
  \label{eq:Glamu}
  G_{\lambda,\mu} \;\coloneqq\; G(\lambda,\mu)
  \;=\; \lambda^2 G_0 \;+\; \lambda\mu\,G_1 \;+\; \mu^2 G_2,
  \qquad (\lambda,\mu)\in\Fq^2.
\end{equation}
All three matrices are symmetric over $\Fq$ by the symmetry of the trace form. {We also have $$ G_{\lambda,\mu}=(\lambda I + \mu A_L)^{\top}G_0 (\lambda I + \mu A_L),$$ where $A_L$ is the matrix associated to the $\F_q$-linear operator $L$ with respect to the $\F_q$-basis $\{ e_1,\ldots,e_m\}$. Then $\det(G_{\lambda,\mu})=\det(G_0) \det(\lambda I + \mu A_L) ^2$.}
Theorem~\textup{\ref{thm:master}} gives immediately
\[
  \dim_{\Fq}\Hull(\Clm) \;=\; \rank(\phiop) \;-\; \rank(G_{\lambda,\mu}).
\]
{In particular, when $\phiop$ is bijective  then $C_{\lambda,\mu}$ is LCD.}
The hull strata $\mathcal{S}_h$ of Definition~\textup{\ref{def:strata}} restrict, for
$(\lambda,\mu)\in\Fq^2\setminus\{(0,0)\}$, to subsets of $\PP^1(\Fq)$
(since $\Fq^*$-scaling on $\Fq^2$ gives the usual projective line over $\Fq$), which
we continue to denote $\mathcal{S}_h$ in this context.  For $(\lambda,\mu)\in\Fqm^2$, the hull
strata on $\PP^1(\Fqm)$ are described by the adjoint formulation of
Section~\textup{\ref{sec:frobenius}}.

Since $G_{\lambda,\mu}$ is homogeneous of degree two in $(\lambda,\mu)$, the hull
dimension (for $(\lambda,\mu)\in\Fq^2$) is a function of the projective ratio
$\rho=\lambda/\mu\in\Fq$.
Setting $\mu=1$ without loss of generality, the analysis reduces to the
\emph{monic pencil}
\begin{equation}
  \label{eq:pencil}
  \widetilde{G}(\rho) \;\coloneqq\; \rho^2 G_0 \;+\; \rho\,G_1 \;+\; G_2, \qquad \rho\in\Fq,
\end{equation}
with the degenerate fibre at $\rho=\infty$ (i.e.\ $\mu=0$) giving
$G_{\lambda,0}=\lambda^2 G_0$ of rank $\rank(G_0)$. {In the case $\mu=0$, $C_{\lambda,0}$ is LCD and hence we can assume $\mu \ne 0$ and consider the map $\phi_{\lambda / \mu, 1}$. Then $$\det(G_{\lambda / \mu,1})=\det(G_0)\det(\lambda / \mu I + A_L)^2,$$ where the solutions of the  equation $\det(\gamma I + A_L)=0$ are exactly the eigenvalues of the $\F_q$-linear operator $\tilde{L}(x)=-L(x)$.}

\begin{definition} 
\label{def:discriminant}
The \emph{discriminant} of the pencil is
\begin{equation}
  \label{eq:discriminant}
  \Delta(\rho) \;\coloneqq\; \det\!\bigl(\widetilde{G}(\rho)\bigr)
  \;=\; \det(\rho^2 G_0 + \rho G_1 + G_2) \;\in\; \Fq[\rho],
\end{equation}
a polynomial of degree at most $2m$, with leading coefficient $\det(G_0)$ when
$G_0$ is invertible.  {Here ``discriminant'' means the determinant of the quadratic Gram pencil.  Since
$G_{\rho,1}=(\rho I+A_L)^{\top}G_0(\rho I+A_L)$ and $G_0$ is nonsingular,
}\[
  {\Delta(\rho)=\det(G_0)\det(\rho I+A_L)^2.
}\]
{Consequently, for $\rho\in\Fq$, $\Delta(\rho)=0$ if and only if the defining operator $\phi_{\rho,1}$ is singular.  Thus the discriminant detects singular fibres; on such fibres the LCD property must still be decided from $\rank(\phi_{\rho,1})-\rank(G_{\rho,1})$.
}\end{definition}

\begin{proposition}
\label{prop:lcdiff}
{Let $(\lambda:\mu)\in\PP^1(\Fq)$.  If $\mu=0$, then $\phi_{\lambda,0}=\lambda I$ }is bijective and {$C_{\lambda,0}$ is LCD.  If $\mu\ne0$ and $\rho=\lambda/\mu$, then
}\[
{\Delta(\rho)=0 }\iff {\phi_{\rho,1}\text{ is singular}.
}\]
{Hence every bijective point is LCD, whereas at a root of $\Delta$ one has
}\[
{\dim_{\Fq}}\Hull{(C_{\rho,1})=\dim_{\Fq}\ker G_{\rho,1}-\dim_{\Fq}\ker\phi_{\rho,1}.
}\]
{In particular, the distinct roots of $\Delta$ form the singular locus of the base-field pencil and number at most $m$.}
\end{proposition}

\begin{proof}
{The factorization $G_{\rho,1}=(\rho I+A_L)^{\top}G_0(\rho I+A_L)$ gives
$\Delta(\rho)=\det(G_0)\det(\rho I+A_L)^2$.  Since $G_0$ is nonsingular, the determinant vanishes exactly when $\rho I+A_L$ is singular, equivalently when $\phi_{\rho,1}$ is singular.  The hull formula follows from Theorem~}\textup{\ref{thm:master}} {after writing ranks as $m$ minus nullities.  Although $\Delta$ has degree $2m$, its distinct roots are those of the degree-$m$ polynomial $\det(\rho I+A_L)$, so there are at most $m$ of them}.
\end{proof}

When $G_0$ is invertible, setting $A=G_0^{-1}G_1$ and $B=G_0^{-1}G_2$ gives
\begin{equation}
  \label{eq:factored}
  \Delta(\rho) \;=\; \det(G_0)\cdot\det(\rho^2 I_m + \rho A + B),
\end{equation}
whose roots are the generalized eigenvalues of the companion linearization
\begin{equation}
  \label{eq:companion}
  \mathcal{L}(\rho) \;\coloneqq\;
  \rho\begin{pmatrix}I&0\\0&I\end{pmatrix}
  -\begin{pmatrix}0&I\\-B&-A\end{pmatrix}
  \;\in\;\mathcal{M}_{2m\times 2m}(\Fq[\rho]).
\end{equation}
The Kronecker structure of $\mathcal{L}$ encodes the complete rank behavior of the pencil.

The pencil simplifies substantially when $\ochar(\Fq)=2$.

\begin{proposition}
\label{prop:char2vanish}
If $\ochar(\Fq)=2$ and $L$ is self-adjoint \textup{(}i.e.\ $L^\dagger=L$\textup{)}, then $G_1=0$, so
$G_{\lambda,\mu}=\lambda^2 G_0+\mu^2 G_2$.
\end{proposition}

\begin{proof}
From~\eqref{eq:G1}, $(G_1)_{ij}=\Tr(e_i\,L(e_j))+\Tr(L(e_i)\,e_j)$.
Since $L$ is self-adjoint, $\Tr(e_i\,L(e_j))=\langle e_i,L(e_j)\rangle
=\langle L^\dagger(e_i),e_j\rangle=\langle L(e_i),e_j\rangle=\Tr(L(e_i)\,e_j)$,
so $(G_1)_{ij}=2\,\Tr(L(e_i)\,e_j)=0$ in characteristic~$2$.
\end{proof}


\begin{corollary}
\label{cor:char2hull}
Suppose $\ochar(\Fq)=2$ and $L$ is self-adjoint.  Then $\dim_{\Fq}\Hull(\Clm)$ (for $(\lambda,\mu)\in\Fq^2$) depends on
$(\lambda,\mu)$ only through $(\lambda^2,\mu^2)$, the discriminant satisfies
$\Delta(\rho)=\det(\rho^2 G_0+G_2)\in\Fq[\rho^2]$, and the hull dimension is
constant on cosets of $(\Fq^*)^2\times(\Fq^*)^2$ in $\Fq^*\times\Fq^*$.
\end{corollary}

\begin{proof}
Since $G_1=0$ by Proposition~\ref{prop:char2vanish}, the pencil collapses to
$G_{\lambda,\mu}=\lambda^2 G_0+\mu^2 G_2$,
which depends on $(\lambda,\mu)$ only through $(\lambda^2,\mu^2)$.  Setting $G_1=0$
in~\eqref{eq:discriminant} yields $\Delta(\rho)=\det(\rho^2 G_0+G_2)$, manifestly
a polynomial in $\rho^2$.  The coset claim is immediate.
\end{proof}

We illustrate the structure with the smallest non-trivial instance.

\begin{remark}
The operator $L(x)=x^{q^k}$ is self-adjoint over $\Fqm$ if and only if $m\mid 2k$,
since $L^\dagger(x)=x^{q^{m-k}}$ and $L^\dagger=L$ iff $q^{m-k}\equiv q^k\pmod{q^m-1}$,
i.e.\ $m\mid 2k$.  In particular, $L(x)=x^2=x^{q^1}$ is self-adjoint over $\mathbb{F}_4$
\textup{(}$m=2$, $k=1$, $m\mid 2k$\textup{)}, so Proposition~\textup{\ref{prop:char2vanish}}
applies to the example below.
\end{remark}
\begin{example}
\label{ex:F4}
Let $q=2$, $m=2$, and $L(x)=x^2$.  Write $\mathbb{F}_4=\mathbb{F}_2(\alpha)$ with
$\alpha^2+\alpha+1=0$, and take $\{e_1,e_2\}=\{1,\alpha\}$.  Since
$\Tr_{\mathbb{F}_4/\mathbb{F}_2}(x)=x+x^2$, direct computation gives
$\Tr(1)=0$, $\Tr(\alpha)=1$, $\Tr(\alpha^2)=1$, so
\[
  G_0 \;=\; \begin{pmatrix}0&1\\1&1\end{pmatrix}.
\]
By Proposition~\textup{\ref{prop:char2vanish}}, $G_1=0$.
Since $L(1)=1$ and $L(\alpha)=\alpha^2=\alpha+1$, the same trace evaluations give
$G_2=G_0$.  Hence $G_{\lambda,\mu}=(\lambda^2+\mu^2)G_0=(\lambda+\mu)^2G_0$ in
characteristic~$2$.  As $\det(G_0)=1\ne 0$ over $\mathbb{F}_2$,
\[
  \rank(G_{\lambda,\mu}) \;=\;
  \begin{cases} 2 & \text{if }\lambda\neq\mu,\\ 0 & \text{if }\lambda=\mu.\end{cases}
\]
If $\lambda\neq\mu$, then $\phi_{\lambda,\mu}$ is bijective, so
$\dim_{\mathbb{F}_2}\Hull(C_{\lambda,\mu})=2-2=0$ and the code is LCD.
If $\lambda=\mu\neq 0$, then
$\phi_{\lambda,\lambda}(x)=\lambda(x+x^2)$ has kernel $\mathbb{F}_2$ and hence rank~$1$;
since $G_{\lambda,\lambda}=0$, Theorem~\textup{\ref{thm:master}} gives
$\dim_{\mathbb{F}_2}\Hull(C_{\lambda,\lambda})=1$.
Thus $C_{\lambda,\lambda}$ is self-orthogonal of dimension~$1$, not the whole ambient space.
The discriminant is $\Delta(\rho)=\det((\rho^2+1)G_0)=(\rho+1)^4$ over $\mathbb{F}_2$,
with unique root $\rho=1$, exactly the non-bijective and self-orthogonal projective point.
\end{example}

The preceding results combine into the following theorem, which collects the complete
hull analysis of the two-parameter family.

\begin{theorem}
\label{thm:grampencil}
Let $L$ be a $q$-polynomial over $\Fqm$, $\{e_1,\ldots,e_m\}$ a $\Fq$-basis of $\Fqm$,
$G_0,G_1,G_2$ as in~\eqref{eq:G0}--\eqref{eq:G2}, and
$\Delta(\rho)=\det(\rho^2G_0+\rho G_1+G_2)\in\Fq[\rho]$.  Then:
\begin{enumerate}[label=\textup{(\roman*)}]
  \item For $(\lambda,\mu)\in\Fq^2\setminus\{(0,0)\}$,
$\dim_{\Fq}\Hull(\Clm)=\rank(\phi_{\lambda,\mu})-\rank(\lambda^2 G_0+\lambda\mu G_1+\mu^2 G_2)$.
\item {For $\mu\ne0$ and $\rho=\lambda/\mu\in\Fq$, $\Delta(\rho)=0$ }if and only if {$\phi_{\rho,1}$ is singular.  Thus every bijective base-field point is LCD, while a singular point is LCD exactly when $\rank(G_{\rho,1})=\rank(\phi_{\rho,1})$}.  For general $(\lambda:\mu)\in\PP^1(\Fqm)$, {where the quadratic base-field expansion need not apply, }the hull is determined by Theorem~\textup{\ref{thm:hullgeneral}} via the adjoint operator.
\item If $\ochar(\Fq)=2$ and $L$ is self-adjoint, then $G_1=0$ and $\dim_{\Fq}\Hull(\Clm)$ depends
on $(\lambda,\mu)\in\Fq^2$ only through $(\lambda^2,\mu^2)$.
\item {Since $G_0$ is the Gram matrix of the non-degenerate trace form, it is always invertible.  Hence $\deg\Delta=2m$, but the distinct roots of $\Delta$ are those of $\det(\rho I+A_L)$ and therefore number at most $m$.  Counted with multiplicity, the roots of $\Delta$ in $\overline{\Fq}$ are the generalized eigenvalues of~\eqref{eq:companion}.}
\end{enumerate}
\end{theorem}

\begin{proof}
Parts~(i), (ii), and~(iii) {follow from }Theorem~\textup{\ref{thm:master}}, Proposition~\ref{prop:lcdiff}, and Corollary~\ref{cor:char2hull}, respectively.
{For~(iv), the two factorizations}
\[
  {\Delta(\rho)=\det(G_0)\det(\rho I+A_L)^2
  =\det(G_0)\det(\rho^2I+\rho A+B)}
\]
{show that $\deg\Delta=2m$, while its distinct roots are those of the degree-$m$ polynomial $\det(\rho I+A_L)$.  The roots of $\det(\rho^2I+\rho A+B)$, counted with multiplicity, are precisely the finite generalized eigenvalues of the companion linearization~\eqref{eq:companion}, proving~(iv).}
\end{proof}

\begin{remark}
Theorem~\textup{\ref{thm:grampencil}} reduces the rank analysis of $\{G_{\lambda,\mu}\}$ to two
tasks: {(a) compute and factor $\Delta(\rho)$ over $\Fq$; and
(b) for each root $\rho_0$ of $\Delta$, compute $\dim\ker\widetilde{G}(\rho_0)$.}
Task~(a) is carried out for $L(x)=x^{q^k}$ in Section~\textup{\ref{sec:frobenius}}.
\end{remark}

\begin{remark}
\label{rem:parallel}
The parallel between the two specializations is now transparent.  In the two-parameter family, for $(\lambda,\mu)\in\Fq^2${, the determinant }$\Delta(\rho)\in\Fq[\rho]$ is a degree-$2m$ {singularity polynomial: its base-field roots are precisely the singular fibres of the operator pencil.  The hull dimension at a singular root is then obtained from the rank difference in Theorem~\textup{\ref{thm:master}}.  Roots computed over }$\overline{\Fq}$ {describe the singular locus after scalar extension, and the base-field roots are obtained by intersection with $\Fq$.  The full parameter }line $\PP^1(\Fqm)$ is handled by the adjoint formula.
In the rank-distance family, the hull is the kernel of the $k\times k$ matrix $M$
over $\Fqm$, whose determinant $\det(M)\in\Fqm$ plays the same role.  Both are
instances of the same Gram-matrix philosophy: the difference lies in which parameter varies
(a projective line versus a fixed generating set) and over which field the Gram
matrix is defined ($\Fq$ versus $\Fqm$).
\end{remark}

Table~\ref{tab:parallel} outlines the key distinctions between $\Clm$ and a generic rank-metric code $\mathcal{C}$.

\begin{table}[ht]
\centering
\renewcommand{\arraystretch}{1.25}
\begin{tabular}{lcc}
\toprule
 & Two-parameter family $\Clm$ & Rank-distance code $\mathcal{C}$ \\
\midrule
Code space      & $\im(\phiop)\subseteq\Fqm$, $\Fq$-linear & $\langle X,F_1,\ldots,F_k\rangle_{\Fqm}$, $\Fqm$-linear \\
Bilinear form   & Trace form $\Tr_{\Fqm/\Fq}$               & Delsarte coefficient pairing \\
Gram matrix     & $G_{\lambda,\mu}=\lambda^2G_0+\lambda\mu G_1+\mu^2G_2$ over $\Fq$ & $M=(M_{jj'})$ over $\Fqm$ \\
Size            & $m\times m$                               & $k\times k$ \\
Hull formula    & $\dim\Hull=\rank(\phiop)-\rank(G_{\lambda,\mu})$ & $\dim\Hull=k-\rank(M)$ \\
LCD criterion   & {$\rank(G_{\lambda,\mu})=\rank(\phiop)$ }& $\det(M)\ne 0$ \\
Self-orthogonal & $G_{\lambda,\mu}=0$                       & $M=0$ \\
Self-dual       & possible                                  & never (Thm.~\ref{thm:rdgeneral}(iv)) \\
\bottomrule
\end{tabular}
\caption{Structural parallel between the two specializations of the general Gram-matrix approach.}
\label{tab:parallel}
\end{table}

\section{The Frobenius twist}
\label{sec:frobenius}

We now carry out the particular case of Theorem~\textup{\ref{thm:grampencil}} for the
canonical choice $L=X^{q^k}$, $1\le k\le m-1$, the \emph{Frobenius twist} by
$q^k$.  

{Let $\phi_{\lambda,\mu}=\lambda I +\mu L$. We can assume $\mu \ne 0$ and consider $\phi_{\lambda / \mu,1}$. Recall that $\det(G_{\lambda / \mu,1})=0$ if and only if $\det(\lambda / \mu I + A_L)=0$, or equivalently $G_{\lambda / \mu,1}$ is singular if and only if $\lambda / \mu$ is an eigenvalue of $\tilde{L}(x)=-x^{q^k}$.}

This choice yields explicit closed-form expressions for $G_0$, $G_1$, $G_2$,
for the discriminant $\Delta(\rho)$ and its roots, and for the hull dimension at every
point of $\PP^1(\Fqm)$.  A secondary motivation for this section is to determine the
\emph{spectrum} of the family $\Clm$, namely the collection $\{N_\delta\}_\delta$ where
\[
  N_\delta \;\coloneqq\; \#\bigl\{(\lambda,\mu)\in\Fq^2\setminus\{(0,0)\}
  \;:\; \dim\Hull(C_{\lambda,\mu})=\delta\bigr\}.
\]
Equivalently, after projectivizing, one may count the points of $\PP^1(\Fq)$ (or of
$\PP^1(\Fqm)$ when working over the larger field of parameters) according to hull
dimension.  In the Frobenius-twist case this spectrum can be described completely, so
this section may also be read as a first step toward the broader problem of whether
such hull spectra distinguish families, detect equivalences, or encode finer structural
information about the codes $C_{\lambda,\mu}$.  Throughout this section we write
$d\coloneqq\gcd(k,m)$ and work with the \emph{normal basis}
$\{e_i=\beta^{q^{i-1}}\}_{i=1}^{m}$ for a fixed normal element $\beta\in\Fqm$
(which exists by the Normal Basis Theorem).

Since $\Tr_{\Fqm/\Fq}(\beta^{q^i}\beta^{q^j})=\Tr(\beta^{q^i+q^j})$ depends only on
$i-j\pmod{m}$, each of $G_0$, $G_1$, $G_2$ is a \emph{circulant matrix} over $\Fq$.
Write $\omega_j=\Tr_{\Fqm/\Fq}(\beta^{1+q^j})$ for $j\in\mathbb{Z}/m\mathbb{Z}$.  The
$(i,j)$-entry of $G_0$ is $\Tr(e_i e_j)=\omega_{j-i}$, so
\begin{equation}
  \label{eq:G0frob}
  G_0 \;=\; \mathrm{circ}(\omega_0,\omega_1,\ldots,\omega_{m-1}).
\end{equation}
Since $L(e_i)=e_i^{q^k}=\beta^{q^{k+i-1}}=e_{i+k}$ (indices mod $m$), the operator
$L$ simply shifts the normal basis by $k$.  Therefore
\[
  \Tr(e_i\,L(e_j)) \;=\; \Tr(e_i\,e_{j+k}) \;=\; \omega_{j+k-i},
\]
and~\eqref{eq:G1} gives $(G_1)_{ij}=\omega_{j+k-i}+\omega_{i+k-j}$, which is
also circulant:
\begin{equation}
  \label{eq:G1frob}
  G_1 \;=\; \mathrm{circ}(\omega_k+\omega_{m-k},\,\omega_{k+1}+\omega_{m-k-1},\,\ldots,\,
  \omega_{k+m-1}+\omega_{1-k}).
\end{equation}
Note that in characteristic~$2$, {the conclusion }$G_1=0$ {follows from }Proposition~\ref{prop:char2vanish} {when the Frobenius twist is self-adjoint, equivalently when $m\mid 2k$.  Characteristic}~$2$ {alone does not force $G_1$ to vanish. }For $G_2$,
\[
  \Tr(L(e_i)L(e_j)) \;=\; \Tr(e_{i+k}\,e_{j+k}) \;=\; \omega_{j-i},
\]
so $G_2=G_0$; the Frobenius twist is an isometry of the trace form.
Substituting into~\eqref{eq:Glamu},
\begin{equation}
  \label{eq:GlamuFrob}
  G_{\lambda,\mu} \;=\; (\lambda^2+\mu^2)G_0 \;+\; \lambda\mu\,G_1,
  \qquad (\lambda,\mu)\in\Fq^2.
\end{equation}

Assume henceforth that $\gcd(m,\operatorname{char}(\Fq))=1$, so that {$X^m-1$ is separable.  Choose a finite splitting field $K=\F_{q^s}$ with $m\mid(q^s-1)$ and let $\xi\in K$ }be a primitive $m$-th root of unity{.  The circulant matrices are diagonalized by the DFT over $K$ (not necessarily over $\Fq$).  Their ranks over $\Fq$ agree with their ranks after extension to $K$, and roots relevant to the original base-field pencil are obtained by intersecting the root set in $K$ with $\Fq$.  Write
}$\hat{\omega}(t)=\sum_{j=0}^{m-1}\omega_j\xi^{jt}$ for the DFT of the first row of
$G_0$.  The normal-basis Gram matrix $G_0$ is invertible~\cite{Lidl97}, so
$\hat{\omega}(t)\ne 0$ for all $t$.  The $t$-th eigenvalue of $G_1$ is
$(\xi^{kt}+\xi^{-kt})\hat{\omega}(t)$, and of $G_{\lambda,\mu}$ (for $(\lambda,\mu)\in\Fq^2$)
is $\hat\omega(t)\cdot[(\lambda^2+\mu^2)+\lambda\mu(\xi^{kt}+\xi^{-kt})]
=\hat\omega(t)(\lambda+\mu\xi^{kt})(\lambda+\mu\xi^{-kt})$.  Setting $\mu=1$ and
$\rho=\lambda$,
\begin{equation}
  \label{eq:DeltaFrob}
  \Delta(\rho) \;=\; \det(G_0)\cdot \prod_{t=0}^{m-1}
  \bigl(\rho^2 + (\xi^{kt}+\xi^{-kt})\rho + 1\bigr),
  \qquad \rho\in\Fq,
\end{equation}
a polynomial of degree $2m$ in $\Fq[\rho]$.  The $t$-th factor has roots
$\rho=-\xi^{kt}$ and $\rho=-\xi^{-kt}$, so
\[
  \mathcal{Z}(\Delta) \;\coloneqq\; \{\rho\in\overline{\Fq}:\Delta(\rho)=0\}
  \;=\; \{-\xi^{kt}:t=0,\ldots,m-1\}\cup\{-\xi^{-kt}:t=0,\ldots,m-1\}.
\]
When $d=1$, both sets equal the full set of $m$-th roots of unity, so
$\mathcal{Z}(\Delta)=\{-\zeta:\zeta^m=1\}$.

The following lemma computes the kernel dimension of $\phiop$; Part~(i) of the next proposition is an immediate consequence.
\begin{lemma} 
\label{lem:kerfrob}
Let $L=X^{q^k}$, $d=\gcd(k,m)$, $(\lambda,\mu)\in\Fqm^2$ with $\mu\ne 0$, and $\rho=\lambda/\mu$.
Then
\[
  \dim_{\Fq}\ker(\phiop)
  \;=\;
  \begin{cases}
    d & \text{if }(-\rho)^{(q^m-1)/(q^d-1)}=1,\\
    0 & \text{otherwise.}
  \end{cases}
\]
When the kernel is non-trivial it is an $\mathbb{F}_{q^d}$-subspace of $\Fqm$ of
$\Fq$-dimension $d$, consisting of $0$ together with a single coset of
$\mathbb{F}_{q^d}^*$ in $\Fqm^*$.
For $\mu=0$, $\phiop(x)=\lambda x$ is bijective (for $\lambda\ne 0$).
\end{lemma}

\begin{proof}
Without loss of generality, normalize $\mu=1$.  A nonzero element $x\in\ker(\phiop)$ then satisfies $x^{q^k-1}=-\rho$.
The map $x\mapsto x^{q^k-1}$ on $\Fqm^*$ has image the unique subgroup $H$ of
$\Fqm^*$ of index $\gcd(q^k-1,q^m-1)=q^d-1$, characterized by
$y^{(q^m-1)/(q^d-1)}=1$.  If $-\rho\notin H$ there are no nonzero solutions.  If
$-\rho\in H$, the preimage of $\{-\rho\}$ is a single coset of $\ker(x\mapsto
x^{q^k-1})=\mathbb{F}_{q^d}^*$ in $\Fqm^*$, which together with $0$ is an
$\mathbb{F}_{q^d}$-subspace of $\Fq$-dimension $d$.
\end{proof}

Next, we provide an LCD criterion on $\PP^1(\Fq)$.
\begin{proposition}
\label{prop:frobenius-hull}
Let $L=X^{q^k}$, $d=\gcd(k,m)$, and $\gcd(m,\operatorname{char}(\Fq))=1$.
For $(\lambda:\mu)\in\PP^1(\Fq)$ with {$\mu\ne0$ and $\rho=\lambda/\mu$}:
\begin{enumerate}[label=\textup{(\roman*)}]
  \item $\phiop$ is bijective if and only if {$(-\rho)^{(q^m-1)/(q^d-1)}\ne1$}. At $\mu=0$, $\phiop(x)=\lambda x$ is bijective.
  \item {$\Delta(\rho)=0$ }if and only if {$\phiop$ is non-bijective.  Consequently every root of $\Delta$ lies on the singular locus, and every bijective point is LCD}.
  \item If $d=1$, {the roots of $\Delta$ in $\Fq$ }are exactly the elements $\rho\in\Fq$ {satisfying }$(-\rho)^m=1${; these are singular parameters, not bijective non-LCD parameters}.
\end{enumerate}
\end{proposition}

\begin{proof}
Part~(i) is Lemma~\ref{lem:kerfrob}{.  Part}~(ii) {follows from
$\Delta(\rho)=\det(G_0)\det(\rho I+A_L)^2$, since $G_0$ is nonsingular}.  For part~(iii), when $d=1$ {the factorization of~\eqref{eq:DeltaFrob} gives }$\mathcal{Z}(\Delta)=\{-\zeta:\zeta^m=1\}$ {over the splitting field; intersection }with $\Fq$ gives the {stated base-field roots}.
\end{proof}

Proposition~\ref{prop:frobenius-hull} identifies {the roots of $\Delta$ with the singular parameters in }$\PP^1(\Fq)$.  We now compute $\dim\ker\widetilde{G}(\rho_0)$ at each {such root; subtracting the operator nullity $d$ then gives the hull dimension}.

For $\rho_0\in\Fq$, define the \emph{frequency multiplicity}
\[
  \nu(\rho_0) \;\coloneqq\;
  \#\bigl\{t\in\{0,\ldots,m-1\} : \rho_0=-\xi^{kt} \text{ or } \rho_0=-\xi^{-kt}\bigr\}.
\]


\begin{theorem}
\label{thm:stratumcount}
Let $L=X^{q^k}${, $d=\gcd(k,m)$, and assume $\gcd(m,\operatorname{char}(\Fq))=1$.  Let }$\rho_0\in\Fq$ be a root of $\Delta${, and write $\rho_0=-\xi^j$ in the splitting field $K$.  Then $\phi_{\rho_0,1}$ is singular with $\dim_{\Fq}\ker\phi_{\rho_0,1}=d$, while
}\[
{\dim_{\Fq}\ker}\widetilde{G}{(\rho_0)=\nu(\rho_0),
}\qquad
{\dim_{\Fq}}\Hull{(C_{\rho_0,1})=\nu(\rho_0)-d,
}\]
{where
}\[
\nu(\rho_0)=\#{\{}t\in\{0,\ldots,m-1{\}}:\rho_0=-\xi^{kt}\text{ or }\rho_0=-\xi^{-kt}\}.
\]
Moreover:
\begin{enumerate}[label=\textup{(\roman*)}]
  \item {if $2j\not\equiv0\pmod m$}, then $\nu(\rho_0)=2d$ and {$\dim\Hull(C_{\rho_0,1})=d$;
  }\item {if $2j\equiv0\pmod m$}, then $\nu(\rho_0)=d$ and {$\dim\Hull(C_{\rho_0,1})=0$}.
\end{enumerate}
{Thus a singular base-field fibre can be either LCD or }non-LCD{; the discriminant detects singularity, while the multiplicity calculation distinguishes the two hull possibilities}.
\end{theorem}
\begin{proof}
{Over the splitting field $K$, the DFT simultaneously diagonalizes the circulant matrices and gives diagonal entries
$d_t=\hat\omega(t)(\rho_0+\xi^{kt})(\rho_0+\xi^{-kt})$}.  Since $G_0$ is {nonsingular, $\hat\omega(t)\ne0$, so the Gram nullity is exactly }$\nu(\rho_0)$.  {Scalar extension preserves rank, hence this is also the nullity over $\Fq$.  }Because {$\rho_0$ is a root of $\Delta$, Proposition~\ref{prop:frobenius-hull} shows that $\phi_{\rho_0,1}$ is singular; Lemma~\ref{lem:kerfrob} gives operator nullity $d$.  Therefore Theorem~\ref{thm:master} yields
$\dim\Hull=\nu(\rho_0)-d$.  The congruences $kt\equiv j\pmod m$ }and {$kt\equiv-j\pmod m$ each have }$d$ solutions{. They are disjoint when $2j\not\equiv0\pmod m$ and coincide when $2j\equiv0\pmod m$, proving the two cases}.
\end{proof}

\begin{remark}
The {exceptional roots with $2j\equiv0\pmod m$ have }$\nu(\rho_0)=d$ {and therefore hull dimension $0$; the remaining base-field roots have }$\nu(\rho_0)=2d$ {and hull dimension $d$}.
\end{remark}

\begin{corollary}
\label{cor:hullvalues}
Let $d=1$ and $\gcd(m,\operatorname{char}(\Fq))=1$.  {At a base-field root $\rho_0=-\xi^j$ of $\Delta$, one has}
\[
{\dim\Hull(C_{\rho_0,1})=
\begin{cases}
0,&2j\equiv0\pmod m,\\
1,&2j\not\equiv0\pmod m.
\end{cases}}
\]
{Every base-field point that is not a root of $\Delta$ is bijective and hence LCD.}
\end{corollary}
\begin{proof}
{This is }Theorem~\ref{thm:stratumcount} with $d=1${, together with Proposition~\ref{prop:frobenius-hull}}.
\end{proof}

The stratum counts above bound the non-LCD locus; the next corollary makes this precise and shows that LCD codes predominate for large $q$.
\begin{corollary}
\label{cor:lcdcount}
Let $L=X^{q^k}$, $d=\gcd(k,m)$, and $\gcd(m,\operatorname{char}(\Fq))=1$.  {The non-LCD points of $\PP^1(\Fq)$ can occur only among the distinct roots of $\Delta$, hence there are at most $m$ of them.  More precisely, if $\rho=-\xi^j\in\Fq$ is a root, then it is non-LCD exactly when $2j\not\equiv0\pmod m$.  Consequently, for fixed $m$ the proportion of LCD points in $\PP^1(\Fq)$ tends to $1$ as $q\to\infty$}.
\end{corollary}
\begin{proof}
{Outside $\mathcal Z(\Delta)$ the operator is bijective by }Proposition~\ref{prop:frobenius-hull}{, hence the image code is LCD.  At a root, Theorem~\ref{thm:stratumcount} gives hull dimension }either $0$ or {$d$, with the latter occurring precisely when $2j\not\equiv0\pmod m$.  Since $\Delta(\rho)=\det(G_0)\det(\rho I+A_L)^2$ and $\det(\rho I+A_L)$ has degree $m$, there are at most $m$ distinct base-field roots.  The density conclusion follows from $|\PP^1(\Fq)|=q+1$}.
\end{proof}

The results so far cover $(\lambda,\mu)\in\Fq^2$.  To extend the hull analysis to the full projective line $\PP^1(\Fqm)$, we compute the adjoint of $\phiop$ explicitly.  No hypothesis on $\gcd(m,\operatorname{char}(\Fq))$ is needed for this.

\begin{lemma} 
\label{lem:adjoint}
For any $(\lambda,\mu)\in\Fqm^2$, the adjoint of $\phiop(x)=\lambda x+\mu x^{q^k}$
with respect to $\Tr_{\Fqm/\Fq}$ is
\begin{equation}
  \label{eq:adjoint}
  \phiop^{\dagger}(y) \;=\; \lambda\,y \;+\; \mu^{q^{m-k}}\,y^{q^{m-k}}.
\end{equation}
\end{lemma}

\begin{proof}
The key identity is that $\Tr(x^{q^k}\cdot y)=\Tr(x\cdot y^{q^{m-k}})$ holds for all
$x,y\in\Fqm$, so
\[
  \Tr(\phiop(x)\cdot y)
  \;=\; \Tr(\lambda xy)+\Tr(\mu x^{q^k}y)
  \;=\; \Tr\!\bigl(x\cdot(\lambda y+\mu^{q^{m-k}}y^{q^{m-k}})\bigr)
  \;=\; \Tr(x\cdot\phiop^{\dagger}(y)).\qedhere
\]
\end{proof}

With the adjoint in hand, the hull admits an intersection-theoretic description valid over all of $\PP^1(\Fqm)$.
\begin{theorem} 
\label{thm:hullgeneral}
For every $(\lambda,\mu)\in\Fqm^2\setminus\{(0,0)\}$,
\begin{equation}
  \label{eq:hullgeneral}
  \Hull(\Clm) \;=\; \im(\phiop)\;\cap\;\ker(\phiop^{\dagger}).
\end{equation}
\end{theorem}

\begin{proof}
By non-degeneracy of $\Tr_{\Fqm/\Fq}$ and Lemma~\ref{lem:adjoint},
$(\im\phiop)^{\perp}
=\{y:\Tr(\phiop(x)y)=0\;\forall x\}
=\{y:\Tr(x\cdot\phiop^{\dagger}(y))=0\;\forall x\}
=\ker(\phiop^{\dagger})$.
Hence $\Hull(\Clm)=\im(\phiop)\cap(\im\phiop)^{\perp}=\im(\phiop)\cap\ker(\phiop^{\dagger})$.
\end{proof}

The adjoint $\phiop^{\dagger}$ has the same form as $\phiop$ with twist $q^{m-k}$ and
scalar $\mu^{q^{m-k}}$ in place of $\mu$; since $\gcd(m-k,m)=d$,
Lemma~\ref{lem:kerfrob} applies to $\phiop^{\dagger}$ with $m-k$ replacing $k$.

{Because the trace form is non-degenerate, an operator and its adjoint have the same rank.  Hence $\phiop$ is bijective if and only if $\phiop^{\dagger}$ is bijective; the asymmetric cases considered in the previous version cannot occur.  For $\mu\ne0$ put $\rho=\lambda/\mu$ and define
}\[
{\varepsilon}\coloneqq\bigl[(-\rho)^{(q^m-1)/(q^d-1)}=1\bigr].
\]
By Lemma~\ref{lem:kerfrob}{, $\varepsilon=1$ exactly when }$\phiop$ {(equivalently $\phiop^\dagger$) is singular, in which case both kernels have $\Fq$-dimension $d$.
}
\begin{theorem}
\label{thm:hullfrob}
Let $L=X^{q^k}$, $d=\gcd(k,m)$, and $(\lambda,\mu)\in\Fqm^2$ with $\mu\ne0$.  Then
\[
{\dim_{\Fq}C_{\lambda,\mu}}={m-\varepsilon d,
}\]
and
\[
\dim_{\Fq}\Hull({C_{\lambda,\mu})}=
\begin{cases}
0,&\varepsilon=0,\\
\delta=\dim_{\Fq}(\im\phiop\cap\ker\phiop{^\dagger),}&{\varepsilon}=1,
\end{cases}
\qquad 0\le\delta\le d.
\]
{In the singular case, $\delta=d$ if and only if $\ker\phiop^\dagger$ is totally isotropic for the trace pairing.
}\end{theorem}

\begin{proof}
Lemma~\ref{lem:kerfrob} {gives the dimension of $\ker\phiop$ and hence that of its image.  }Theorem~\ref{thm:hullgeneral} gives
{$\Hull(C_{\lambda,\mu})=\im\phiop\cap\ker\phiop^\dagger$.  If $\varepsilon=0$, then $\phiop$ is bijective }and {the hull is zero.  If $\varepsilon=1$, both $\ker\phiop$ and $\ker\phiop^\dagger$ have dimension $d$, whence $0\le\delta\le d$.  Finally, since
$\im\phiop=(\ker\phiop^\dagger)^\perp$, equality }$\delta=d$ {is equivalent to
$\ker\phiop^\dagger\subseteq(\ker\phiop^\dagger)^\perp$, which is precisely total isotropy}.
\end{proof}

\begin{remark}
\label{rem:casesummary}
{There are only two rank configurations: the bijective case, which is always LCD}, {and the singular case, in which $\dim\Hull=\delta$ with $0\le\delta\le d$.  The adjoint has the same rank as the original operator, so no mixed bijectivity cases occur}.
\end{remark}

The {singular case }admits the following explicit criterion for the extremal value $\dim\Hull=d$ in terms of a generator of $\ker\phiop^\dagger$.

\begin{proposition}
\label{prop:isotropic}
Let $L=X^{q^k}$ {and assume $\phiop$ is singular (equivalently, }$\phiop^{\dagger}$ {is singular)}.  Let $x_0\in\ker\phiop^{\dagger}\setminus\{0\}$.  Then
$\ker\phiop^{\dagger}=\langle x_0 \rangle_{\mathbb{F}_{q^d}}$,
and this subspace is totally isotropic for
$\Tr_{\Fqm/\Fq}$ if and only if
\begin{equation}
  \label{eq:isocrit}
  \Tr_{\Fqm/\mathbb{F}_{q^d}}(x_0^2)=0.
\end{equation}
Equivalently, all $\Fq$-inner products among elements of the $\mathbb{F}_{q^d}$-subspace
$\ker\phiop^{\dagger}=\langle x_0 \rangle_{\mathbb{F}_{q^d}}$ vanish.  Consequently,
$\dim\Hull(\Clm)=d$ if and only if~\eqref{eq:isocrit} holds.
\end{proposition}

\begin{proof}
By Lemma~\ref{lem:kerfrob} applied to $\phiop^{\dagger}$, the kernel $\ker\phiop^{\dagger}$ is an
$\mathbb{F}_{q^d}$-line, so $\ker\phiop^{\dagger}=\langle x_0 \rangle_{\mathbb{F}_{q^d}}$ for any nonzero
$x_0\in\ker\phiop^{\dagger}$.  Let $y_1=ax_0$ and $y_2=bx_0$ with
$a,b\in\mathbb{F}_{q^d}$.  Using the tower property of the trace,
\[
  \Tr_{\Fqm/\Fq}(y_1y_2)
  = \Tr_{\Fqm/\Fq}(abx_0^2)
  = \Tr_{\mathbb{F}_{q^d}/\Fq}\!\Bigl(ab\,\Tr_{\Fqm/\mathbb{F}_{q^d}}(x_0^2)\Bigr).
\]
Since the pairing $(a,b)\mapsto \Tr_{\mathbb{F}_{q^d}/\Fq}(ab)$ on $\mathbb{F}_{q^d}$ is non-degenerate,
this vanishes for all $a,b\in\mathbb{F}_{q^d}$ if and only if
$\Tr_{\Fqm/\mathbb{F}_{q^d}}(x_0^2)=0$.  The final claim follows from
Theorem~\ref{thm:hullfrob}.
\end{proof}

A particularly clean situation arises when the operator coincides with its own adjoint.

\begin{theorem} 
\label{thm:selfadjoint}
{If $\mu=0$, then $\phiop(x)=\lambda x$ is automatically self-adjoint.  If $\mu\ne0$, then the operator $\phiop(x)=\lambda x+\mu x^{q^k}$ is self-adjoint (i.e., $\phiop=\phiop^{\dagger}$) if and only if $m\mid 2k$ and $\mu\in\mathbb{F}_{q^d}^*$.}
When these conditions hold, the hull is
\[
  \Hull(\Clm)=\im(\phiop){\cap}\ker(\phiop).
\]
The code is LCD if $\phiop$ is bijective; otherwise, $\dim\Hull(\Clm)=\delta$ {with $0\le\delta\le d$, and }$\delta=d$ if and only if $\ker\phiop$ is totally isotropic.
\end{theorem}
\begin{proof}
{If $\mu=0$, then $\phiop=\lambda I=\phiop^\dagger$, so self-adjointness is automatic.  Assume henceforth that $\mu\ne0$.} By Lemma~\textup{\ref{lem:adjoint}}, the adjoint is $\phiop^{\dagger}(y) = \lambda y + \mu^{q^{m-k}}y^{q^{m-k}}$. Matching this against $\phiop(y) = \lambda y + \mu y^{q^k}$ requires the twist exponents to agree modulo $m$, giving $m-k\equiv k\pmod{m}$, which is equivalent to $m\mid 2k$. It also requires the scalar coefficients to match, giving $\mu^{q^{m-k}}=\mu$. Since $m\mid 2k$, we have $\gcd(m-k,m)=\gcd(k,m)=d$, so the fixed points of $x \mapsto x^{q^{m-k}}$ in $\Fqm$ are exactly the elements of $\mathbb{F}_{q^d}$. Thus, $\mu^{q^{m-k}}=\mu$ if and only if $\mu\in\mathbb{F}_{q^d}$.
Assume now that these conditions hold, so $\phiop=\phiop^{\dagger}$. Theorem~\textup{\ref{thm:hullgeneral}} dictates that $\Hull(\Clm) = \im(\phiop)\cap\ker(\phiop^{\dagger})$, which immediately simplifies to $\im(\phiop)\cap\ker(\phiop)$. Because the operator and its adjoint are identical, {the general two-case description }of Theorem~\textup{\ref{thm:hullfrob}} {specializes directly}. 
If the operator is bijective, then $\ker(\phiop)=\{0\}$, implying $\dim\Hull(\Clm)=0$ and the code is LCD. If it is non-bijective, Theorem~\textup{\ref{thm:hullfrob}} {gives }$\dim\Hull(\Clm)=\delta=\dim(\im\phiop\cap\ker\phiop)$ with $0\le\delta\le d$. The extremal condition $\delta=d$ holds if and only if $\ker\phiop$ is totally isotropic.
\end{proof}





When $\phiop$ is self-adjoint, the Gram and adjoint descriptions of the hull coincide{.  The discriminant still detects bijectivity, while LCD-ness at a singular parameter is decided by the intersection $\im(\phiop)\cap\ker(\phiop)$}.
\begin{remark}
\label{rem:discriminant-sa}
When the operator is self-adjoint ({$m\mid2k$}) and we restrict to parameters {$(\rho,1)\in\Fq^2$}, the Gram and adjoint descriptions {agree.  The condition $\Delta(\rho)\ne0$ is equivalent to bijectivity and therefore implies that }$C_{\rho,1}$ is LCD{.  At a root of $\Delta$, however, the operator is singular and LCD-ness is determined by whether $\im(\phi_{\rho,1})\cap\ker(\phi_{\rho,1})$ is trivial}. 
In the specific case where $\operatorname{char}(\Fq)=2$ (and assuming the operator is self-adjoint with respect to the trace), Proposition~\textup{\ref{prop:char2vanish}} gives $G_1=0$. The Gram matrix simplifies to $G_{\rho,1}=(\rho+1)^2G_0$, and the discriminant becomes $\Delta(\rho)=\det(G_0)\cdot(\rho+1)^{2m}$, which has $\rho=1$ as its unique zero in $\Fq$. 
For general parameters over $\PP^1(\Fqm)$, the {base-field discriminant does not apply; in the self-adjoint case the hull is determined by }$\im(\phiop)\cap\ker(\phiop)$ as in Theorem~\textup{\ref{thm:selfadjoint}}.
\end{remark}
\begin{remark}
\label{rem:consistency}
More generally, for any $(\lambda,\mu)\in\Fq^2$, the Gram approach and the adjoint approach are always consistent. Indeed, since $(\im\phiop)^{\perp}=\ker\phiop^{\dagger}$ (Theorem~\textup{\ref{thm:hullgeneral}}), the identity $\rank G_{\lambda,\mu}=\rank\phiop-\dim\Hull(\Clm)$ from Theorem~\textup{\ref{thm:master}} perfectly matches $\dim\Hull(\Clm)=\dim(\im\phiop\cap\ker\phiop^{\dagger})$. Note that while the Gram formula~\eqref{eq:GlamuFrob} requires $(\lambda,\mu)\in\Fq^2$, the adjoint formula~\eqref{eq:hullgeneral} holds universally for all $(\lambda,\mu)\in\Fqm^2$.
\end{remark}





We close the section with a worked {example.  The general statements above are proved algebraically; the SageMath code is not used as a proof of those statements.  It performs an exhaustive verification for the single field $\F_{64}$ and serves to check the corrected formulas and to display the resulting strata explicitly}.

\begin{example}
\label{ex:frob-q4m3k1}
Let $q=4$, $m=3$, $k=1$, so $d=1$ and the ambient field is $\mathbb{F}_{64}$.
Consider $\phi_{\rho,1}(x)=\rho x+x^4$ for $\rho\in\mathbb{F}_{64}$.
Remark~\ref{rem:casesummary} predicts $\dim\Hull\in\{0,1\}$.

\medskip\noindent\textit{Bijectivity.}
By Lemma~\textup{\ref{lem:kerfrob}}, $\phi_{\rho,1}$ is non-bijective iff
$(-\rho)^{(q^m-1)/(q^d-1)}=(-\rho)^{21}=1$ in $\mathbb{F}_{64}^*$.
Since $\operatorname{char}=2$, $-1=1$, so this becomes $\rho^{21}=1$.  The elements
satisfying this form the unique subgroup of order $21$ in $\mathbb{F}_{64}^*$ (which
has order $63=3\times 21$).  Every $\rho\in\mathbb{F}_4^*$ has $\rho^3=1$, so
$\rho^{21}=(\rho^3)^7=1$: all three elements of $\mathbb{F}_4^*$ give non-bijective
$\phi_{\rho,1}${; equivalently, $\varepsilon=1$}.  The remaining $21-3=18$ non-bijective points lie
in $\mathbb{F}_{64}\setminus\mathbb{F}_4$.

\medskip\noindent\textit{Full computation over $\PP^1(\mathbb{F}_{64})$.}
{An exhaustive SageMath computation over this fixed field }(see Appendix) gives
\[
  |\PP^1(\mathbb{F}_{64})|\;=\;65,\qquad
  |\mathcal{S}_0|\;=\;60,\qquad
  |\mathcal{S}_1|\;=\;5.
\]
All five non-LCD points {are singular and }have $\dim\Hull=1$.  For each:
\[
  \rank(\phi_{\rho,1})=2,\qquad
  \rank(G_{\rho,1})=1,\qquad
  \dim\Hull(C_{\rho,1})=1.
\]
By Theorem~\textup{\ref{thm:hullfrob}}, $\delta=d=1$ {for these five parameters }because
$\ker\phi_{\rho,1}^{\dagger}$ is totally isotropic (verified computationally).
The remaining $21-5=16$ non-bijective projective points all have $\dim\Hull=0$,
{and their hull is trivial; equivalently, the relevant intersection with $\ker\phi^\dagger$ is zero}.

The five non-LCD parameters are
\[
  \rho\in\{a^5+a^4+a^2+1,\;a^4+a^2+a+1,\;a^3+a^2+a,\;a^5+a,\;a^3+a^2+a+1\},
\]
where $a$ is a fixed generator of $\mathbb{F}_{64}^*$.  Among the cube roots of unity
$\{a^3+a^2+a,\;a^3+a^2+a+1,\;1\}$, only the first two give $\dim\Hull=1$; at $\rho=1$
we have {a singular operator }but $\ker\phi_{1,1}^{\dagger}$ is not totally
isotropic, giving $\dim\Hull=0$.  This confirms that the non-LCD locus in
$\PP^1(\mathbb{F}_{64})$ is strictly larger than the set of non-trivial cube roots of
unity, and that the Gram-pencil discriminant criterion on $\PP^1(\mathbb{F}_4)$
does not extend verbatim to $\PP^1(\mathbb{F}_{64})$.
\end{example}

\begin{remark}
\label{rem:quantum}
The {hull information obtained here can be translated into entanglement consumption only after fixing a standard EAQECC construction and the relevant classical duality.  For a classical $[n,k]$ code $C\subseteq\Fq^n$ with Euclidean dual and parity-check matrix $H$, the usual parity-check construction uses
}\[
{c=}\rank{(HH^{\top})=n-k-\dim}\Hull{(C)
}\]
{ebits.  Equivalently, the generator Gram rank satisfies $\rank(GG^{\top})=k-\dim\Hull(C)$ and appears in the dual version of the construction.  Thus the hull determines a complementary Gram rank rather than being, by itself, }the ebit count{.  Accordingly, our hull formulas may be inserted into these identities whenever the code under consideration is represented in a coordinate model compatible with the Euclidean duality used by the chosen EAQECC construction.  We do not identify the trace-dual image codes or the Delsarte-dual }rank-distance {codes with a specific EAQECC without this additional choice}.
\end{remark}

\section{Conclusions and open problems}
\label{sec:conclusions}

In this paper we developed a unified Gram-matrix approach for computing the hull of two families of linearized polynomial codes over $\Fqm$.  The master hull--rank theorem (Theorem~\textup{\ref{thm:master}}) reduces the {image-code }problem to a rank comparison between an operator and its associated Gram matrix{.  The factorization $G_T=A_T^{\top}G_0A_T$ shows, in particular, that every bijective }image-code {operator gives an LCD code and that the determinant of the base-field Gram pencil detects singularity, not non-LCD behavior on the bijective locus}.  For the Frobenius twist $L=X^{q^k}$, {diagonalization over a splitting field gives an explicit discriminant factorization; at a base-field root the operator has nullity $d$ }and the {hull dimension is $0$ or $d$.  Over $\PP^1(\Fqm)$, the adjoint formula reduces the singular case to }$\im\phiop\cap\ker\phiop^{\dagger}$, with the extremal value $d$ characterized by Proposition~\textup{\ref{prop:isotropic}}.

Several natural directions remain open.  It would be valuable to extend the Gram-matrix
approach beyond the two-parameter family to more general $\Fqm$-linear combinations of
$q$-polynomials, where the structure matrices $\Gamma_{ij}$ are no longer circulant and
closed-form discriminants may not exist.  A second direction is to determine, for
prescribed hull dimension $h$, whether the Frobenius family contains members
that are simultaneously MRD in the rank metric; the interplay between the rank-distance
bound and the hull constraint is not yet understood.  Finally, the quantum coding
interpretation (Remark~\textup{\ref{rem:quantum}}) raises the question of whether the
families studied here can be used to construct EA-QECCs meeting a Gilbert--Varshamov-type
bound with prescribed ebit count, extending the classical result of Sendrier~\cite{Sendrier04}
to the entanglement-assisted setting.

\bigskip
\noindent
\textbf{Acknowledgements.}  
{The authors express their deep appreciation to the editor, as well as to the anonymous referee, whose thorough reading and constructive comments have greatly improved the paper.} 
The third-named author (PS) would like to thank  the first-named author (DB) for the invitation at the Dipartimento di Matematica e Informatica  at Universit\`a degli Studi di Perugia, Italy, and the great working conditions while this paper was being written. The first-named author (DB) and second-named author (GGG) thank the Italian National Group for Algebraic and Geometric Structures and their Applications (GNSAGA—INdAM) which supported the research.
The authors would like to thank Prof. Dibyendu Roy for help with the SageMath code.

\appendix
\section{SageMath code for Example~\ref{ex:frob-q4m3k1}}

\begin{lstlisting}
# =============================================================================
# Verification of Example~\ref{ex:frob-q4m3k1}
# Parameters: q=4, m=3, k=1, L(x) = x^{q^k} = x^4
# Field: F_{q^m} = GF(4^3) = GF(64)
#
# Key results:
#   - G_1 != 0  (char 2 does NOT force G_1 = 0; that requires m | 2k)
#   - Hull computed via rank(phi) - rank(G), G[i][j] = Tr(phi(e_i)*phi(e_j))
#   - Also cross-checked via adjoint: Hull = im(phi) cap ker(phi^dag)
#   - Result: 60 LCD points, 5 non-LCD points (all hull dim = 1)
# =============================================================================

from sage.all import *
from collections import Counter

q = 4
m = 3
k = 1

# -----------------------------------------------------------------------
# Field setup
# -----------------------------------------------------------------------
F.<a> = GF(q^m, modulus='primitive')
p = F.characteristic()                  # 2
r = ZZ(q).exact_log(p)                  # 2  (q = p^r)
Fq = GF(q)                              # abstract GF(4) for matrix entries

# GF(4) as subfield of GF(64): elements satisfying x^q = x
Fq_in_F = [x for x in F if x^q == x]
alpha_q  = next(x for x in Fq_in_F if x != F(0) and x != F(1) and x^(q-1) == F(1))

print("=" * 60)
print(f"Field: GF({q}^{m}) = GF({q^m})")
print(f"p = {p}, r = {r}, q = p^r = {q}")
print(f"gcd(m, char) = gcd({m}, {p}) = {gcd(m, p)}  (coprime: {gcd(m, p) == 1})")
print(f"|F_{{q^m}}*| = {q^m - 1} = {factor(q^m - 1)}")
print(f"GF(4) elements inside GF(64): {Fq_in_F}")
print(f"GF(4) primitive element alpha_q = {alpha_q}")
print()

# -----------------------------------------------------------------------
# Relative trace  Tr_{F/Fq}
# -----------------------------------------------------------------------
def Tr(z):
    return sum(z^(q^i) for i in range(m))   # result lies in GF(4) inside F

# -----------------------------------------------------------------------
# Normal basis for F/Fq.
# beta is a normal element iff {beta, beta^q, ..., beta^{q^{m-1}}} are
# GF(q)-linearly independent.  We check this via the (m*r) x (m*r) matrix
# over GF(p) whose columns are the GF(p)-coordinate vectors of the m conjugates
# together with their alpha_q-multiples.
# -----------------------------------------------------------------------
def is_normal(beta):
    conjs  = [beta^(q^i) for i in range(m)]
    vecs   = [vector(GF(p), beta_i._vector_()) for beta_i in conjs]
    avecs  = [vector(GF(p), (alpha_q * beta_i)._vector_()) for beta_i in conjs]
    return matrix(GF(p), vecs + avecs).rank() == m * r

beta         = next(b for b in F if b != F(0) and is_normal(b))
normal_basis = [beta^(q^i) for i in range(m)]

print(f"Normal element beta = {beta}")
print(f"Normal basis: {normal_basis}")
print()

# -----------------------------------------------------------------------
# GF(q)-coordinate extraction via the change-of-basis matrix.
# We build the (m*r) x (m*r) = 6x6 GF(p)-matrix M whose columns are
# [e_0, e_1, e_2, alpha_q*e_0, alpha_q*e_1, alpha_q*e_2] in GF(p)^6.
# Its inverse gives the GF(q)-coordinates of any element of F.
# -----------------------------------------------------------------------
def gfp_vec(x):
    return vector(GF(p), x._vector_())

M6 = matrix(GF(p),
            [gfp_vec(e)        for e in normal_basis] +
            [gfp_vec(alpha_q*e) for e in normal_basis]).transpose()   # 6x6
M6_inv = M6.inverse()

def coords_Fq(x):
    """Return [c0, c1, ..., c_{m-1}] in GF(q) (inside F) with sum c_i*e_i = x."""
    y = M6_inv * gfp_vec(x)     # length m*r = 6, entries in GF(p)
    result = []
    for i in range(m):
        a0 = y[i]               # GF(p) component
        a1 = y[i + m]           # alpha_q component
        result.append(F(a0) + F(a1) * alpha_q)
    return result

# Verify on basis elements
for j in range(m):
    c = coords_Fq(normal_basis[j])
    assert all(c[i] == (F(1) if i == j else F(0)) for i in range(m)), \
        f"coords_Fq failed on e_{j}"
print("coords_Fq verified on all basis elements.")

# -----------------------------------------------------------------------
# Map GF(q)-elements-as-F-elements to abstract Fq = GF(4) for matrices.
# Build an explicit embed dict: {0->0, 1->1, alpha_q->gen, alpha_q^2->gen^2}
# -----------------------------------------------------------------------
Fq_gen = Fq.gen()
embed = {
    F(0):       Fq(0),
    F(1):       Fq(1),
    alpha_q:    Fq_gen,
    alpha_q^2:  Fq_gen^2,
}

def to_Fq(x):
    """Convert a GF(4)-element x (living inside F) to abstract GF(4)."""
    return embed[x]

# -----------------------------------------------------------------------
# phi_{lam,mu}(x) = lam*x + mu*x^{q^k}
# -----------------------------------------------------------------------
def phi(lam, mu, x):
    return lam * x + mu * x^(q^k)

# Matrix of phi in the normal basis, as a matrix over abstract Fq = GF(4)
def phi_matrix(lam, mu):
    cols = [coords_Fq(phi(lam, mu, normal_basis[j])) for j in range(m)]
    return matrix(Fq, m, m, lambda i, j: to_Fq(cols[j][i]))

# Gram matrix G[i][j] = Tr(phi(e_i)*phi(e_j)), entries in Fq
def G_matrix(lam, mu):
    return matrix(Fq, m, m,
                  lambda i, j: to_Fq(Tr(phi(lam, mu, normal_basis[i]) *
                                         phi(lam, mu, normal_basis[j]))))

def hull_dim(lam, mu):
    r_phi = phi_matrix(lam, mu).rank()
    r_G   = G_matrix(lam, mu).rank()
    return r_phi - r_G

# -----------------------------------------------------------------------
# Adjoint: phi^dag(y) = lam*y + mu^{q^{m-k}} * y^{q^{m-k}}
# -----------------------------------------------------------------------
def phi_dag(lam, mu, y):
    return lam * y + mu^(q^(m-k)) * y^(q^(m-k))

def hull_via_adjoint(lam, mu):
    """Dimension of im(phi) cap ker(phi^dag) -- independent verification."""
    im_phi  = set(phi(lam, mu, x) for x in F)
    ker_dag = set(y for y in F if phi_dag(lam, mu, y) == F(0))
    inter   = [z for z in im_phi & ker_dag if z != F(0)]
    if not inter:
        return 0
    vecs = matrix(GF(p), [gfp_vec(z) for z in inter])
    return vecs.rank() // r

# -----------------------------------------------------------------------
# Structure matrices for reference
# -----------------------------------------------------------------------
G0 = matrix(Fq, m, m, lambda i, j: to_Fq(Tr(normal_basis[i] * normal_basis[j])))
G1 = matrix(Fq, m, m,
            lambda i, j: to_Fq(Tr(normal_basis[i] * normal_basis[j]^(q^k)) +
                                 Tr(normal_basis[i]^(q^k) * normal_basis[j])))
G2 = matrix(Fq, m, m, lambda i, j: to_Fq(Tr(normal_basis[i]^(q^k) * normal_basis[j]^(q^k))))

omega = [Tr(beta^(1 + q^j)) for j in range(m)]
print(f"omega = {omega}")
print(f"\nG_0 =\n{G0}")
print(f"det(G_0) = {G0.det()}  (nonzero: {G0.det() != 0})")
print(f"\nG_1 =\n{G1}")
print(f"G_1 is zero: {G1.is_zero()}")
print(f"  (G_1 = 0 iff m | 2k;  {m} | {2*k}: {(2*k) % m == 0})")
print(f"\nG_2 =\n{G2}")
print(f"G_2 == G_0: {G2 == G0}")

# -----------------------------------------------------------------------
# Compute strata over PP^1(F)
# -----------------------------------------------------------------------
print()
print("=" * 60)
print(f"HULL STRATA OVER PP^1(GF({q^m}))")
print("=" * 60)

strata = Counter()
strata[hull_dim(F(1), F(0))] += 1    # rho = infty
for lam in F:
    strata[hull_dim(lam, F(1))] += 1

for h in sorted(strata):
    print(f"  S_{h}: {strata[h]} projective points")
print(f"  Total: {sum(strata.values())} = q^m + 1 = {q^m + 1}")

# -----------------------------------------------------------------------
# Non-LCD points: detailed report with adjoint cross-check
# -----------------------------------------------------------------------
print()
print("=" * 60)
print("NON-LCD POINTS  (hull dim > 0)")
print("=" * 60)

exp_nonbij = (q^m - 1) // (q^gcd(k, m) - 1)

non_lcd = [(lam, F(1)) for lam in F if hull_dim(lam, F(1)) > 0]
if hull_dim(F(1), F(0)) > 0:
    non_lcd.append((F(1), F(0)))

for (lam, mu) in non_lcd:
    h     = hull_dim(lam, mu)
    h_adj = hull_via_adjoint(lam, mu)
    rphi  = phi_matrix(lam, mu).rank()
    rG    = G_matrix(lam, mu).rank()
    rho   = lam / mu if mu != F(0) else "infty"
    nonbij = (mu != F(0)) and ((-lam / mu)^exp_nonbij == F(1))

    print(f"rho = {rho}:")
    print(f"  rank(phi) = {rphi},  rank(G) = {rG},  hull = {h}  [adjoint: {h_adj}]")
    print(f"  non-bijective: {nonbij}")
    print(f"  G =")
    for row in G_matrix(lam, mu).rows():
        print(f"    {list(row)}")
    print()

# -----------------------------------------------------------------------
# m-th roots of unity
# -----------------------------------------------------------------------
print("=" * 60)
print(f"m-th roots of unity in GF({q^m})*  (rho^m = 1):")
m_roots = [x for x in F if x != F(0) and x^m == F(1)]
print(f"  {m_roots}  (count: {len(m_roots)})")
for rho in m_roots:
    print(f"  rho = {rho}:  hull dim = {hull_dim(rho, F(1))}")

# -----------------------------------------------------------------------
# Summary
# -----------------------------------------------------------------------
print()
print("=" * 60)
print("SUMMARY")
print("=" * 60)
print(f"  PP^1(GF({q^m})) has {q^m + 1} projective points.")
print(f"  S_0 (LCD):  {strata[0]} points")
for h in sorted(strata):
    if h > 0:
        print(f"  S_{h}:        {strata[h]} points  (hull dim = {h})")
print()
print(f"  G_1 = 0 iff m | 2k:  {m} | {2*k} is {(2*k) % m == 0}  =>  G_1 != 0.")
print(f"  All ranks computed over GF({q}) using the normal basis.")
print(f"  G[i][j] = Tr(phi(e_i)*phi(e_j)) directly -- no quadratic formula assumed.")
print(f"  Adjoint hull = im(phi) cap ker(phi^dag) verified at every non-LCD point.")

\end{lstlisting}

\end{document}